\documentclass[11pt]{article}
\usepackage[colorlinks]{hyperref}
\usepackage{amsmath}
\usepackage{amsthm}
\usepackage{thmtools}
\usepackage{amssymb}
\usepackage{graphicx}
\usepackage{multicol}
\usepackage{multirow}
\usepackage{color}
\usepackage{tikz}
\usetikzlibrary{arrows.meta,calc}
\definecolor{ppszblue}{RGB}{62,102,139}
\definecolor{ppszred}{RGB}{176,64,56}
\definecolor{ppszgray}{RGB}{105,105,105}
\tikzset{
  ppsz picture/.style={font=\small, line cap=round, line join=round},
  ppsz vertex/.style={circle, draw=black, fill=white, line width=.4pt,
    inner sep=0pt, minimum size=3.6pt, outer sep=0pt},
  ppsz edge/.style={draw=black, line width=.48pt},
  ppsz retained/.style={draw=ppszblue, line width=.75pt},
  ppsz deleted/.style={draw=ppszred, line width=.75pt,
    dash pattern=on 2.2pt off 1.6pt},
  ppsz label/.style={inner sep=1pt},
  ppsz kept strip/.style={draw=ppszblue!55, line width=.42pt, line cap=butt},
  ppsz cut strip/.style={draw=ppszred, line width=1.05pt, line cap=butt},
  ppsz column/.style={fill=black!7, draw=none},
  ppsz tree edge/.style={draw=black, line width=.42pt},
  ppsz stub/.style={draw=black!45, line width=.4pt},
  ppsz pair on/.style={draw=ppszblue, line width=.7pt,
    dash pattern=on .6pt off 1.5pt, -{Stealth[length=3.6pt,width=3pt]},
    shorten <=2pt, shorten >=2.2pt},
  ppsz pair off/.style={draw=ppszred, line width=.7pt,
    dash pattern=on .6pt off 1.5pt, -{Stealth[length=3.6pt,width=3pt]},
    shorten <=2pt, shorten >=2.2pt},
  ppsz region/.style={fill=black!5, draw=black!25, line width=.35pt,
    rounded corners=6pt},
}

\usepackage[letterpaper,margin=1in,bottom=1in]{geometry}
\usepackage[capitalize]{cleveref}
\usepackage[utf8]{inputenc}
\usepackage[english]{babel}
\usepackage{mathtools}
\usepackage{needspace}

\newcommand{\E}{\mathbb{E}}

\usepackage{interval}
\intervalconfig{soft open fences}

\newtheorem{theorem}{Theorem}[section]

\newtheorem{lemma}[theorem]{Lemma}
\newtheorem{corollary}[theorem]{Corollary}
\newtheorem{proposition}[theorem]{Proposition}

\usepackage{tabularx}
\usepackage{booktabs}
\usepackage{longtable}
\usepackage{threeparttable}
\usepackage{adjustbox}

\newtheorem{inputlemma}[theorem]{Imported estimate}
\newtheorem{importedtheorem}[theorem]{Imported theorem}

\newcommand{\PPSZ}{\textnormal{\textsc{PPSZ}}}
\newcommand{\Pw}{\mathcal P^{(w)}}
\newcommand{\TwoCC}{\mathrm{TwoCC}}
\newcommand{\ID}{\mathrm{ID}}
\newcommand{\eps}{\varepsilon}
\newcommand{\gain}{\mathrm{gain}}
\newcommand{\Thr}{\mathrm{Thr}}
\newcommand{\Lreg}{L_{\mathrm{reg}}}
\newcommand{\Lirr}{L_{\mathrm{irr}}}
\newcommand{\Preg}{P_{\mathrm{reg}}}
\newcommand{\hbin}{h_2}
\newcommand{\sat}{\mathrm{sat}}
\newcommand{\Prb}{\mathbb P}
\newcommand{\Fbase}{F^{(3)}}
\newcommand{\Faug}{\widetilde F}
\newcommand{\Grep}{\Gamma_{\mathrm{avg}}}

\hypersetup{pdftitle={A Better Analysis For PPSZ For 3-SAT},pdfauthor={Tao Jiang and Shaowei Cai}}

\begin{document}

\title{A Better Analysis For PPSZ For 3-SAT}
\author{%
  Tao Jiang \qquad Shaowei Cai\\[1.2ex]
  \normalsize Key Laboratory of System Software (Chinese Academy of Sciences)\\
  \normalsize State Key Laboratory of Computer Science\\
  \normalsize Institute of Software, Chinese Academy of Sciences\\
  \normalsize School of Computer Science and Technology\\
  \normalsize University of Chinese Academy of Sciences\\
  \normalsize Beijing, China\\[0.8ex]
  \normalsize\texttt{\char`\{jiangt,caisw\char`\}@ios.ac.cn}}
\date{}
\maketitle

\begin{abstract}
The PPSZ algorithm of Paturi, Pudl\'ak, Saks, and Zane solves $3$-SAT by processing the variables in a uniformly random order, setting a variable whenever a bounded number of clauses implies its value, and guessing it otherwise. Its running time is determined by how many variables are forced in this process, and the exact exponent of this simple algorithm is a central question in the study of exponential-time satisfiability algorithms. Scheder showed that the classical analysis of PPSZ is not tight for $3$-SAT. We ask how much further a sharper analysis of the unmodified algorithm can go.

Our analysis uses auxiliary distributions on variable orderings that correlate the positions of variables sharing critical clauses. Such correlations make forcing more likely for some variables, can make it less likely for others, and incur an entropy cost. We choose the correlated pairs at random on short components, so that every pair is selected with the same probability, and we express the forcing loss as a sum of weights attached to these pairs in advance. Gains and losses then scale by the same factor and can be compared directly. Combining the resulting estimate with the complementary estimate of Scheder through an explicit linear-programming certificate gives a uniform improvement for every uniquely satisfiable formula.

We prove that the original PPSZ algorithm, with a sufficiently large but fixed implication strength, solves Unique-$3$-SAT in randomized time $O^*(1.306969598^n)$. Through the lifting framework of Scheder and Steinberger, it solves general $3$-SAT in randomized time $O^*(1.307031578^n)$. Both bounds improve the guarantees obtained from the analysis of Scheder for the same algorithm. These improvements follow entirely from the analysis, while PPSZ continues to use a uniformly random variable order and unbiased guesses. It shows that the known guarantees for PPSZ can be strengthened by analysis alone.
\end{abstract}


\section{Introduction}\label{sec:intro}

Satisfiability of 3-CNF formulas is the canonical NP-complete problem, and determining how fast it can be solved in the worst case has been a central question in exact exponential-time algorithms for several decades. Since no subexponential algorithm is believed to exist, the question is quantitative: for which constants $c<2$ can a satisfiable formula on $n$ variables be solved in time $O^*(c^n)$, where $O^*$ suppresses factors polynomial in the input size? The strongest answers to this question are randomized, and two ideas have shaped most of them. The first is local search. Sch\"oning~\cite{Schoening1999} showed that a random walk on assignments, started from a uniformly random assignment and restarted after a linear number of steps, finds a satisfying assignment of a 3-CNF formula in expected time $O^*((4/3)^n)$. The second is the satisfiability coding lemma of Paturi, Pudl\'ak, and Zane~\cite{PPZ1999}, which relates the logical structure of a formula to short descriptions of its satisfying assignments. The two ideas succeed on different kinds of inputs: local search benefits from many satisfying assignments, whereas the coding lemma is strongest when satisfying assignments are isolated. This complementarity explains why several later algorithms combine them. Throughout, a 3-CNF formula may have clauses of length at most three, and Unique-$3$-SAT is the promise problem in which the input has at most one satisfying assignment.

The coding lemma rests on a simple observation. If a satisfying assignment $\alpha$ is isolated in some direction, then flipping a variable $x$ falsifies a clause in which $x$ is the only true literal under $\alpha$; such a clause is called critical for $x$. Whenever the other variables of this clause are set before $x$, their values determine the value of $x$. The algorithm of Paturi, Pudl\'ak, and Zane (PPZ) exploits this by setting the variables in a random order, inferring a variable from a unit clause when possible, and guessing it otherwise. PPSZ~\cite{PPSZ} strengthens the inference step: a variable is inferred whenever its value is implied by a bounded number of clauses of the current residual formula, which allows chains of reasoning rather than a single unit clause. The resulting algorithm remains strikingly simple, and its analysis underlies much of the subsequent progress on the problem.

We now describe PPSZ more precisely. It processes the variables in a uniformly random order. When a variable is reached, the algorithm sets it if some set of at most $w$ clauses of the current residual formula implies its value; otherwise, it sets the variable to an unbiased random bit. Write $\PPSZ_w$ for the algorithm with implication strength $w$. Fix a satisfying assignment $\alpha$. A run outputs $\alpha$ if every guess agrees with $\alpha$, since inferred values agree with $\alpha$ automatically. If $G(\pi)$ counts the guesses along the run consistent with $\alpha$ under the order $\pi$, then for uniquely satisfiable formulas
\[
 \Prb[\PPSZ(F)=\alpha]=\E_{\pi}\bigl[2^{-G(\pi)}\bigr].
\]
Each variable that is forced rather than guessed therefore saves a factor of two, and the running time of PPSZ is governed by the probability that a variable is forced at the moment it is processed. Bounding this forcing probability from below is the core of every analysis of the algorithm. The classical analysis shows that, as $w$ grows, every variable of a uniquely satisfiable 3-CNF formula is forced with probability at least $s_3-\eps_w$, where $s_3=2-2\ln2$ and $\eps_w\to0$. This gives the exponent $p_0=1-s_3=2\ln2-1$ and the base $2^{p_0}=1.3070319\ldots$ for Unique-$3$-SAT.

Subsequent work has developed PPSZ in several directions. One line combines PPSZ with local search or with other procedures, so that each component handles the inputs on which the other is weak~\cite{IwamaTamaki2004,Rolf2006,HMS2011}. A second line modifies PPSZ itself: a slight modification extends the unique-case bound to general formulas~\cite{Hertli2011}, and other modifications treat sparse and dense formulas differently~\cite{Hertli2014,QinWatanabe2020} or bias the guesses according to the formula~\cite{HKZZ2019,QinWatanabe2022}. A third line asks how well PPSZ can possibly perform and constructs formulas on which its success probability is provably small~\cite{CSTT2013,PST2017,ST2020}. The fourth line, to which this paper belongs, leaves the algorithm unchanged and sharpens the analysis of its success probability~\cite{SchederSteinberger}. Within this line, Scheder~\cite{SchederFOCS,SchederJournal} showed that the classical bound is not tight: by exploiting how critical clauses overlap, the forcing probabilities of the original algorithm can be bounded more sharply.

The distinction between improving the algorithm and improving its analysis is important here. The running-time behavior of the original PPSZ algorithm is also of independent theoretical interest. This algorithm uses no information about the input beyond the implications that it observes, and its behavior is governed by a clean probabilistic process. Determining its exponent is therefore a question about this process rather than about algorithm design. The improved bounds of Scheder motivate a closer examination of the forcing process in the original PPSZ algorithm.

This paper addresses the following question: can the running-time guarantee of the original PPSZ algorithm be improved further through a sharper analysis, without modifying the algorithm? Our analysis retains the uniform variable ordering, unbiased guessing rule, and bounded-implication mechanism of the original PPSZ algorithm, applied to the input formula. All improvements come from a more refined accounting of when variables are forced. The additional randomness that appears in the analysis is a proof device and does not affect any run of the algorithm.

Answering this question requires controlling several competing effects. The analysis of Scheder rests on a change of measure~\cite{SchederFull,SchederJournal}. Instead of reasoning directly about the uniform order, one may describe the order by independent placements in $[0,1]$ and evaluate the number of forced variables under any auxiliary placement law $D$. The success probability under the uniform order is then at least $2^{-n+\E_D[\mathrm{Forced}]-\mathrm{KL}_2(D\Vert U)}$, where $U$ is the uniform law. The auxiliary law may favor placements that make forcing likely, but every deviation from uniformity is paid for by its relative entropy.

The useful deviations are correlations. A variable $x$ is forced, roughly, when the variables of its critical clause appear before it, and the same holds recursively for the variables that force those. If the two other variables $y$ and $z$ of a critical clause $(x\vee\bar y\vee\bar z)$ tend to be placed early together, then $x$ is forced more often; this is a forcing reward. The same correlation, however, also changes the conditional placement of $y$ or $z$ wherever these variables occur deeper in the implication trees of other variables, and there it can make forcing less likely; this is a forcing loss. Correlating many overlapping pairs also creates large dependent structures whose forcing probabilities and entropy are difficult to control. A sharper analysis must therefore choose the correlated pairs so that rewards, losses, and entropy can all be bounded under the same auxiliary law and compared with one another.

A second difficulty is uniformity. A correlation-based estimate benefits from variables with several critical clauses, but it weakens when many variables have indegree zero or one in the critical-clause digraph. Scheder handles these regimes with two complementary estimates, a regular one and an irregular one. A guarantee for all uniquely satisfiable formulas requires combining such estimates without knowing in advance which regime the input belongs to.

\subsection{Our results}

Our results hold at a sufficiently large but fixed implication strength, independent of $n$. The analysis uses virtual clauses derived from the input; Section~\ref{sec:inputs} shows that they can be simulated by increasing the strength by a constant factor, so the algorithm is always run on the original formula.

\begin{samepage}
\begin{theorem}[Unique-$3$-SAT]\label{thm:main}
Let
$\gamma_{\mathrm{new}}=0.0000687793$.
There exist a finite implication strength $w_0$ and an integer $n_0$ such that, for every $w\ge w_0$, every $n\ge n_0$, and every 3-CNF formula $F$ on $n$ variables with unique satisfying assignment $\alpha$,
\[
 \Prb[\PPSZ_w(F)=\alpha]
 \ge 2^{-p_0n+\gamma_{\mathrm{new}}n}
 =2^{-n+s_3n+\gamma_{\mathrm{new}}n}.
\]
Consequently, independent repetition gives a randomized algorithm with running time
\[
 O^*(1.306969598^n).
\]
\end{theorem}
\end{samepage}

Theorem~\ref{thm:main} is the direct output of our analysis. The guarantee is uniform over uniquely satisfiable inputs, and PPSZ never needs to determine the structural quantities used in the proof. For general formulas, we apply the lifting theorem of Scheder and Steinberger~\cite{SchederSteinberger}, which converts a unique-case improvement into a smaller improvement for all satisfiable formulas. Section~\ref{sec:lifting} gives a quantitative version of this framework at a fixed finite strength.

\begin{samepage}
\begin{corollary}[General $3$-SAT via Scheder--Steinberger]\label{cor:general}
There exist a finite implication strength $w_1$ and an integer $n_1$ such that, for every $w\ge w_1$, every $n\ge n_1$, and every satisfiable 3-CNF formula $F$ on $n$ variables,
\[
 \Prb[\PPSZ_w(F)\text{ succeeds}]
 \ge 2^{-p_0n+\eta n},
 \qquad
 \eta=0.000000364.
\]
Consequently, general $3$-SAT can be solved by repeated runs of the original PPSZ algorithm in randomized time
\[
 O^*(1.307031578^n).
\]
\end{corollary}
\end{samepage}

The following table compares these bounds with those obtained from the analysis of Scheder for the same algorithm. The bases are rounded outward. The general-case entry for that analysis is obtained by applying the same lifting calculation to its published unique-case bonus $\gamma_{\mathrm{old}}=1/15218$.
\begin{center}
\small
\renewcommand{\arraystretch}{1.15}
\begin{tabular}{@{}lcc@{}}
\toprule
analysis of original PPSZ & Unique-$3$-SAT & general $3$-SAT\\
\midrule
Scheder & $O^*(1.306972377^n)$ & $O^*(1.307031594^n)$\\
this work & $O^*(1.306969598^n)$ & $O^*(1.307031578^n)$\\
\bottomrule
\end{tabular}
\end{center}
These bounds show that the correlations in the forcing process admit further analytical refinement, with gains that persist after combining the regular and irregular estimates and extending the result to general $3$-SAT. The table compares upper bounds for the original PPSZ algorithm, whose exact exponent for $3$-SAT remains unknown.

\subsection{Overview of the analysis}

Following Scheder, we encode correlations by a \emph{sibling multigraph}. Each variable $x$ has a canonical critical clause $(x\vee\bar y\vee\bar z)$, and this clause contributes an edge between $y$ and $z$. A subgraph $L$ of this multigraph specifies which pairs of placements are correlated. Our analysis combines three ideas, which address the difficulties described above in turn.

The first idea concerns the shape of the correlated structure. Short components allow us to control forcing probabilities and entropy through local estimates. We therefore correlate along a subgraph $L$ whose components have at most 17 edges. Rather than fixing one such subgraph, we choose $L$ at random from a candidate graph $G$ of maximum degree two, so that every edge of $G$ belongs to $L$ with the same probability $\rho=17/19$. On each path or cycle of $G$, this is achieved by a uniformly rotated or shifted deletion pattern followed by thinning (Figure~\ref{fig:equal-marginal}). Parallel edges are kept or deleted together, so every correlated pair, and not only every edge instance, receives the same marginal. The constant $17/19$ is the best possible retention for components with at most 17 edges, as the 19-cycle shows.

The random choice of $L$ enters the analysis only through averaging. For each fixed selection $L$, we bound rewards, losses, and entropy under one auxiliary placement law determined by $L$. Each such bound is a valid lower bound on the success probability of the original algorithm, and hence so is their average. Averaging is useful only if rewards and losses scale in the same way. Rewards are indexed by the selected edges and scale with $\rho$ automatically. For the losses, we show that the forcing loss of each root $x$ is at most a sum of nonnegative weights, one for each selected correlated pair, computed from fixed implication trees before $L$ is chosen (Figure~\ref{fig:edge-linear-loss}). Equal marginals then scale the expected loss by the same factor $\rho$ as the reward. Entropy is bounded for each fixed law on its short components before any averaging takes place.

The third idea addresses uniformity. The resulting regular estimate and the irregular estimate of Scheder are both affine in three structural coordinates of the input: the fractions of variables of indegree zero and one, and the fraction of variables with two critical clauses in a one-round augmentation. Each estimate is weak for some values of these coordinates, but the two are weak in different places. We bound the larger of the two estimates from below, uniformly over all values of the coordinates, by a linear program. An explicit feasible dual solution weights the two estimates so that one coordinate cancels and the others have nonnegative coefficients, which yields the bonus of Theorem~\ref{thm:main} with a strict margin. This margin absorbs the error caused by working at a finite implication strength.

The rest of the paper is organized as follows. Section~\ref{sec:inputs} defines the objects of the analysis: weak implication at a fixed strength, the virtual completion of the input, critical clauses and their trees, the structural coordinates, and the candidate graph $G$. It also states the change-of-measure inequality and the inherited irregular estimate. Section~\ref{sec:combination} states the averaged regular estimate, reduces both estimates to affine bounds, and proves Theorem~\ref{thm:main} from the dual certificate. Section~\ref{sec:lifting} states the quantitative lifting result and derives Corollary~\ref{cor:general}. Appendix~\ref{app:regular-proof} contains the proofs behind the regular estimate: the equal-marginal selection, the placement laws and conditioning bounds, the edge-linear forcing loss, the rewards, the entropy bound, and the uniform finite-strength argument. Appendix~\ref{app:arithmetic} records the numerical certificate, and Appendix~\ref{app:lifting-proof} proves the lifting proposition.

\subsection{Related work}

The algorithms that combine PPSZ with local search exploit the complementarity noted above. Iwama and Tamaki~\cite{IwamaTamaki2004} showed that the two algorithms, which are favorable for formulas with few and with many satisfying assignments respectively, can be combined into an algorithm for general $3$-SAT that is faster than either of them, although running them independently does not suffice for this purpose. Rolf~\cite{Rolf2006} sharpened the analysis of this combination. Hertli, Moser, and Scheder~\cite{HMS2011} refined the PPSZ side by using variables that take the same value in all satisfying assignments, and showed that this refinement fits into such combinations. These results improved the general-case bound while the PPSZ component kept its classical unique-case guarantee. Hertli~\cite{Hertli2011} then showed that a slightly modified PPSZ algorithm attains this unique-case bound for general $3$-SAT, which gave a faster algorithm for general $3$-SAT than these combinations.

Improvements in the unique case have come both from changes to the algorithm and from changes to the analysis. Hertli~\cite{Hertli2014} obtained the first bound below $2^{p_0n}$ for Unique-$3$-SAT by an algorithm that treats sparse and dense formulas differently: it uses PPSZ unchanged when many clauses are critical, a separate algorithm when the formula is sparse, and a modified PPSZ that exploits non-critical clauses otherwise. Qin and Watanabe~\cite{QinWatanabe2020} improved this algorithm. Hansen, Kaplan, Zamir, and Zwick~\cite{HKZZ2019} introduced Biased-PPSZ, in which the guesses are biased according to the formula, and obtained faster algorithms for $k$-SAT; Qin and Watanabe~\cite{QinWatanabe2022} refined this approach for $3$-SAT. Scheder~\cite{SchederFOCS,SchederJournal} showed that improvements of a similar kind hold for the original algorithm for every $k\ge3$. His analysis replaces the uniform order by an auxiliary placement law through a change of measure, correlates the placements of variables along the sibling graph, and combines a regular and an irregular estimate. Our work builds directly on this framework. We use the detailed $3$-SAT estimates of the full version~\cite{SchederFull}; in particular, we use its irregular estimate at a different admissible value of its parameter. Scheder and Steinberger~\cite{SchederSteinberger} simplified the analysis of Hertli and extended it to general $k$-SAT and constraint satisfaction problems. Their lifting theorem, which transfers a unique-case improvement to all satisfiable formulas by distinguishing whether many or few variables can still take both values when they are processed, is the tool that we use to derive Corollary~\ref{cor:general}.

Lower bounds clarify how much any analysis can hope to gain. The literature distinguishes a weak version of PPSZ, which infers a value when a small set of clauses implies it, from a strong version, which infers a value derivable by resolution of bounded width. The algorithm $\PPSZ_w$ analyzed in this paper is the weak version; as Section~\ref{sec:finite-strength} notes, a weak implication from $w$ clauses also has a resolution derivation of width at most $3w$. Chen, Scheder, Talebanfard, and Tang~\cite{CSTT2013} constructed satisfiable $k$-CNF formulas on which strong PPSZ has savings $O(\log^2k/k)$, giving the first exponential lower bounds for the algorithm. Pudl\'ak, Scheder, and Talebanfard~\cite{PST2017} gave uniquely satisfiable hard instances, with savings at most $(2+\eps)/k$ for weak PPSZ and $O(\log k/k)$ for strong PPSZ, and Scheder and Talebanfard~\cite{ST2020} obtained savings at most $(1+\eps)\cdot2/k$ for strong PPSZ with small resolution width. These constructions describe the behavior of PPSZ as $k$ grows; they do not determine the exponent for $k=3$, which is the setting of this paper.

Within this landscape, our contribution concerns the analysis of the original algorithm. Beyond the inherited framework, we introduce the equal-marginal short-component selection, the edge-linear form of the forcing loss together with the averaging of fixed-selection certificates, a refined entropy bound, the affine combination of the regular and irregular estimates with an explicit dual certificate, and a quantitative finite-strength specialization of the lifting framework.

\section{Preliminaries}\label{sec:inputs}

\subsection{Weak implication and fixed strength}\label{sec:finite-strength}

The proof heuristic $\Pw$ infers $x=b$ when some set of at most $w$ residual clauses implies this value. It is sound and monotone under restrictions. For fixed $w$, a run of $\PPSZ_w$ takes $n^{O(w)}$ time and polynomial space. Throughout, the algorithm is applied to the original input $F$.

We use the error notion of Definition~1.8 in~\cite{SchederSteinberger}. It measures the failure to infer a frozen value along a consistent run, averaged over a uniform ordering, and it is bounded uniformly in the formula, the satisfying assignment, and the variable. Their Theorem~1.10, which restates the PPSZ estimate, gives
\begin{equation}\label{eq:pw-error}
 p_w=p_0+\eps_w,\qquad \eps_w\ge0,\qquad
 \eps_w\longrightarrow0\quad(w\to\infty),
\end{equation}
where $\eps_w$ is independent of $F$ and $n$.

For each fixed choice of the analytical parameters, we write the regular and irregular approximation losses as
$\xi_X(w)n+r_{X,w}(n), X\in\{R,I\}$,
with $\xi_X(w)\to0$ and a uniform remainder $r_{X,w}(n)=o(n)$ for each fixed $w$. Appendix~\ref{app:uniform} chooses the tree heights uniformly over all formulas and all allowed selections $L$. In particular, all finite-height losses can be included in $\xi_X(w)n$, so that we may take $r_{X,w}(n)=0$. This uniformity allows a single size threshold to be chosen after the strength has been fixed.

A weak implication from $w$ clauses of a 3-CNF involves at most $3w$ variables and has a resolution derivation of width at most $3w$. Each clause in the virtual augmentation below has at most two original supporting clauses. Consequently, weak strength $u$ on the virtual formula is simulated by weak strength $2u$ on the original formula, and also by resolution width $6u$. The simulation is part of the analysis and requires no preprocessing step in PPSZ.

\subsection{Critical clauses and clause trees}\label{sec:structure}

We assume $n\ge4$ in the structural argument; the thresholds in the theorems absorb the cases of bounded size. In a virtual copy of $F$, discard tautologies and repeated literals. Replace each clause $C$ of length $\ell<3$ by all $2^{3-\ell}$ sign extensions on $3-\ell$ existing variables outside $C$. The resulting exact-3 formula $\Fbase$ has the same variables and satisfying assignments as $F$, and each added clause is supported by one original clause.

We normalize the unique satisfying assignment to the all-ones assignment. As in Definition~31 of the full version~\cite{SchederFull}, define the one-round augmentation
\[
 \Faug=\Fbase\cup\{C:\ C\text{ is a non-tautological 3-clause},\
 B\models C\text{ for some }B\subseteq\Fbase,\ |B|\le2\}.
\]
All clauses in this construction use distinct variables. A clause is critical for $x$ if its only satisfied literal is $x$. Define $\TwoCC$ to be the set of variables with at least two distinct critical clauses in $\Faug$. The augmentation is performed once, without iterating the closure. Lemma~\ref{lem:virtual-simulation} proves that the implication simulation is preserved under restrictions, and the discussion following it explains why a pair-implied critical clause of length less than three places its positive variable in $\TwoCC$ when $n\ge4$.

For each variable $x$, choose a canonical critical clause $C_x=(x\vee\bar y\vee\bar z)$ from $\Fbase$; uniqueness guarantees its existence. The critical-clause digraph contains the arcs $x\to y,x\to z$. The sibling multigraph instead contains an edge instance $e_x=\{y,z\}$ labeled by $x$. The sibling multigraph is therefore loopless and has $n$ edge instances, and the degree of each vertex equals its indegree in the digraph. Let $J_i$ be the full indegree-$i$ class and put
\[
 \ID_i=J_i\setminus\TwoCC,\qquad
 i_0=\frac{|\ID_0|}{n},\quad i_1=\frac{|\ID_1|}{n},\quad
 \tau=\frac{|\TwoCC|}{n}.
\]
The full version uses $\ID_i$ for the full classes in Section~6 and for the classes outside $\TwoCC$ in Section~8. We use $J_i$ and $\ID_i$, respectively, to keep this distinction explicit.

If two canonical clauses $C_x$ and $C_w$ have the same negative literals, then $e_x$ and $e_w$ are parallel instances of one support pair. Their clause rewards are counted separately, while the pair receives a single correlation term.

The full canonical clause trees are fixed before a subgraph is selected. The root of $T_x$ has label $x$. At a node, flip the labels on its root path, choose a violated clause of $\Fbase$, and create a child labeled $z$ for each negative literal $\bar z$ in that clause. Use the canonical clause of the node label whenever it is violated; otherwise, use a fixed lexicographic choice. Thus a node labeled $a$ that uses $C_a$ has the two endpoints of $e_a$ as its children. No label repeats on a root path, so every such tree is finite. A node is canonical if all nodes on its root path lie outside $\TwoCC$ and use their canonical clauses. Write $\operatorname{Can}(T_x)$ for the canonical nodes, $a(u)$ for the label of $u$, and $d(u)$ for its depth. Truncation to a fixed finite height is deferred to the probability analysis.

The label densities measure the weighted occurrence of a label in these fixed trees. Define
\begin{align*}
 \operatorname{LD}(b,T_a;r)
 &=\frac{(1-2r)^2}{(1-r)^3}
   \sum_{\substack{v\in\operatorname{Can}(T_a)\\a(v)=b}}r^{d(v)+1},\\
 \operatorname{LD}(b,T_a)&=\int_0^{1/2}\operatorname{LD}(b,T_a;r)\,dr.
\end{align*}
The full-tree quantities and the classifications below depend on the input and the canonical choices, but not on $L$ or on a sampled placement. The uniform tail bound in Appendix~\ref{app:uniform} relates them to the finite trees used at fixed implication strength.

\subsection{The candidate graph}\label{sec:candidate-edges}

Using Lemma~34 of~\cite{SchederFull}, we choose an edge-instance set $H$ of maximum degree two with $|H|\ge n-|J_1|-2|J_0|$. Consequently,
\begin{equation}\label{eq:H-lower}
 |H|\ge n-|\ID_1|-2|\ID_0|-2|\TwoCC|.
\end{equation}
Let $H_{\rm free}$ consist of the instances in $H$ whose endpoints lie outside $\TwoCC$ and whose sibling pair is not generated by any canonical clause labeled by a variable in $\TwoCC$. The label restriction concerns all canonical clauses, regardless of whether their instances belong to $H$. This global free-pair convention gives the conservative bound
\begin{equation}\label{eq:free-count}
 |H_{\rm free}|\ge |H|-4|\TwoCC|.
\end{equation}
To count the excluded instances, we charge each variable in $\TwoCC$ for its incident edges in $H$ and for all instances of its canonical sibling pair. The degree bound permits at most two instances of each type, giving a total charge of at most four per variable.

Fix $\Thr>0$. We put every instance of a free pair $\{y,z\}$ in $H_{\rm high}$ if at least one of its directional label densities is at least $\Thr$, and choose a dense endpoint for each such pair. Let $H_{\rm rest}$ contain the non-high instances of $H_{\rm free}$ incident to any chosen endpoint. Since $H$ has maximum degree two, $|H_{\rm rest}|\le|H_{\rm high}|$. The candidate graph
$G=H_{\rm free}\setminus(H_{\rm high}\cup H_{\rm rest})$
satisfies
\begin{equation}\label{eq:candidate-count}
 |G|\ge |H|-2|H_{\rm high}|-4|\TwoCC|.
\end{equation}
Every pair in $G$ has both directional densities below $\Thr$. Each designated dense endpoint of a high pair lies outside the support of every $L\subseteq G$; this property is used in the high-edge reward estimate.

For a multigraph $L$, write $\underline L$ for its simple support. Distinct edge instances have distinct canonical-clause labels and therefore contribute separately to the rewards. Correlation terms and forcing-loss sums use each pair of $\underline L$ once. In a parallel-edge cycle, Lemma~\ref{lem:edge-marginals} selects the two instances together, giving the same marginal to the instances and their support pair.

We select $L\subseteq G$ with at most 17 edges per component so that each component admits a controlled placement density. Lemma~\ref{lem:edge-marginals} supplies a finite distribution with equal edge-instance and support-pair marginals, and Figure~\ref{fig:equal-marginal} shows its construction. Every component of $G$ is a path, a cycle, or a parallel two-cycle, and the law is a product over components. On a component $C$, a uniformly rotated or shifted deletion pattern deletes each edge in the same number of equally likely outcomes, giving a common retention probability $q_C\ge\rho$. Keeping the resulting candidate with probability $\rho/q_C$ lowers this common value to exactly $\rho$. The 19-cycle is the extremal case $q_C=\rho$, in which no thinning occurs.
\begin{figure}[htbp]
 \centering
\begin{tikzpicture}[ppsz picture]
  \pgfmathsetmacro{\Rout}{1.80}   
  \pgfmathsetmacro{\rin}{0.52}    
  \pgfmathsetmacro{\dr}{0.0655}   
  \pgfmathsetmacro{\hstep}{360/19}

  \begin{scope}
    \fill[ppsz column] ({90-\hstep/2}:\rin-.07)
      -- ({90-\hstep/2}:\Rout)
      arc[start angle={90-\hstep/2}, end angle={90+\hstep/2}, radius=\Rout]
      -- ({90+\hstep/2}:\rin-.07)
      arc[start angle={90+\hstep/2}, end angle={90-\hstep/2}, radius=\rin-.07]
      -- cycle;
    \foreach \s in {0,...,18}{
      \pgfmathsetmacro{\r}{\rin+\s*\dr}
      \foreach \j in {0,...,18}{
        \pgfmathsetmacro{\a}{90+\hstep/2-\j*\hstep-1.6}
        \pgfmathsetmacro{\b}{90-\hstep/2-\j*\hstep+1.6}
        \pgfmathtruncatemacro{\cut}{(\j==\s) || (\j==mod(\s+10,19))}
        \ifnum\cut=1
          \draw[ppsz cut strip] (\a:\r) arc[start angle=\a, end angle=\b, radius=\r];
        \else
          \draw[ppsz kept strip] (\a:\r) arc[start angle=\a, end angle=\b, radius=\r];
        \fi
      }
    }
    \draw[ppsz edge] (0,0) circle[radius=\Rout];
    \foreach \i in {0,...,18}
      \node[ppsz vertex] at ({90+\hstep/2-\i*\hstep}:\Rout) {};
    \node[ppsz label] at (90:\Rout+.24) {$e$};
    \node[ppsz label, font=\footnotesize] at (0,-\Rout-.56)
      {(a) $19$-cycle, rotated $\{0,10\}$};
    \node[ppsz label, font=\footnotesize] at (0,-\Rout-1.0)
      {$\Pr[e\notin L]=\frac{2}{19}$};
  \end{scope}

  \begin{scope}[xshift=2.75cm, yshift=1.55cm]
    \pgfmathsetmacro{\W}{5.20}
    \pgfmathsetmacro{\de}{\W/26}
    \pgfmathsetmacro{\rowgap}{0.143}
    \pgfmathsetmacro{\top}{-0.34}
    \fill[ppsz column] (12*\de,.14) rectangle (13*\de,{\top-18*\rowgap-.18});
    \foreach \r in {0,...,18}{
      \pgfmathsetmacro{\y}{\top-\r*\rowgap-(\r==18)*0.09}
      \foreach \j in {0,...,25}{
        \pgfmathtruncatemacro{\cut}{(\r==18) || (mod(\j,18)==\r)}
        \ifnum\cut=1
          \draw[ppsz cut strip] ({\j*\de+.022},\y) -- ({(\j+1)*\de-.022},\y);
        \else
          \draw[ppsz kept strip] ({\j*\de+.022},\y) -- ({(\j+1)*\de-.022},\y);
        \fi
      }
    }
    \draw[ppsz edge] (0,0) -- (\W,0);
    \foreach \i in {0,...,26} \node[ppsz vertex] at (\i*\de,0) {};
    \node[ppsz label] at (12.5*\de,.3) {$e$};
    \node[ppsz label, font=\scriptsize, anchor=east] at (-.05,\top) {$r=0$};
    \node[ppsz label, font=\scriptsize, anchor=east] at (-.05,{\top-17*\rowgap}) {$r=17$};
    \node[ppsz label, font=\scriptsize, anchor=east] at (-.05,{\top-18*\rowgap-.09}) {$\varnothing$};
    \node[ppsz label, font=\footnotesize] at (\W/2,-3.91)
      {(b) path, residue mod $18$, thinned};
    \node[ppsz label, font=\footnotesize] at (\W/2,-4.35)
      {$\Pr[e\notin L]=\frac1{18}\cdot\frac{18}{19}+\frac1{19}=\frac{2}{19}$};
  \end{scope}

  \begin{scope}[xshift=10.95cm]
    \foreach \s in {0,...,18}{
      \pgfmathsetmacro{\r}{\rin+\s*\dr}
      \ifnum\s<2
        \draw[ppsz cut strip] (87:\r) arc[start angle=87, end angle=-87, radius=\r];
        \draw[ppsz cut strip] (93:\r) arc[start angle=93, end angle=267, radius=\r];
      \else
        \draw[ppsz kept strip] (87:\r) arc[start angle=87, end angle=-87, radius=\r];
        \draw[ppsz kept strip] (93:\r) arc[start angle=93, end angle=267, radius=\r];
      \fi
    }
    \draw[ppsz edge] (0,0) circle[radius=\Rout];
    \node[ppsz vertex] (cy) at (90:\Rout) {};
    \node[ppsz vertex] (cz) at (270:\Rout) {};
    \node[ppsz label, above=2pt] at (cy) {$y$};
    \node[ppsz label, below=2pt] at (cz) {$z$};
    \node[ppsz label] at (180:\Rout+.3) {$e_u$};
    \node[ppsz label] at (0:\Rout+.3) {$e_v$};
    \node[ppsz label, font=\footnotesize] at (0,-\Rout-.56)
      {(c) parallel pair, synchronized};
    \node[ppsz label, font=\footnotesize] at (0,-\Rout-1.0)
      {$\Pr[e_u\notin L]=\Pr[e_v\notin L]=\frac{2}{19}$};
  \end{scope}
\end{tikzpicture}
 \caption{The equal-marginal law of Lemma~\ref{lem:edge-marginals} on three component types. Each ring or row is one outcome on the component, all outcomes in a panel are equally likely, and distinct components are sampled independently. Red marks deleted edges and blue marks retained ones. (a)~Ring $s$, counted from the inside, deletes edges $s$ and $s+10$ modulo 19, leaving paths with 9 and 8 edges. The shaded sector of $e$ meets exactly two deletions, one for each element of the deletion pattern. (b)~Row $r$ deletes the edges with index congruent to $r$ modulo 18. The last row is the empty outcome, which arises when the candidate is discarded with probability $1/19$. (c)~The two instances of a parallel pair are retained or deleted together, so their support pair also has marginal $\rho$. Every cycle with at most 17 edges is treated in the same all-or-nothing way. For a cycle of length $\ell\ge18$, the rotated pattern has $\lceil\ell/18\rceil$ deletions and is kept with probability $\rho/q_\ell$.}
 \label{fig:equal-marginal}
\end{figure}

\Needspace{7\baselineskip}
\subsection{Change of measure and the irregular estimate}

All logarithms in the coefficient functions are natural. We put $f_{\rm KL}(t)=(1-t)\ln(1-t)+t$ for $0\le t\le1$, with $0\ln0=0$.

We represent a variable ordering by placements in $[0,1]$, read in increasing order. Independent uniform placements define the reference law $U$ and induce the uniform variable ordering. The change-of-measure inequality accounts for the entropy cost of an auxiliary placement law. For $D\ll U$, with $\mathrm{Forced}(\pi)$ denoting the number of forced variables along the satisfying run, it gives
\begin{equation}\label{eq:change}
 \E_{\pi\sim U}\!\left[2^{-n+\mathrm{Forced}(\pi)}\right]
 \ge2^{-n+\E_{\pi\sim D}[\mathrm{Forced}(\pi)]-\mathrm{KL}_2(D\|U)}.
\end{equation}
This is Equation~(2) of the full version and Equation~(3) of the journal version~\cite{SchederFull,SchederJournal}. It applies separately to the auxiliary law associated with each fixed graph selection. The simulation above transfers forcing counts obtained from the virtual formula to the original input.

We retain the following irregular estimate of Scheder.
\begin{inputlemma}[Irregular lower bound from Scheder]\label{lem:irregular}
For every fixed $0\le\eps_I\le1/5$ and every admissible finite strength $w$,
\[
 \Prb[\PPSZ_w(F)=\alpha]
 \ge
 2^{-p_0n+\gain_I-\xi_I(w)n-r_{I,w}(n)},
\]
where
\begin{align}
 \gain_I&\ge b_1(\eps_I)|\ID_1|+b_0(\eps_I)|\ID_0|+b_T(\eps_I)|\TwoCC|,
 \label{eq:irregular}\\
 b_1(\eps)&=0.030966\eps-0.0028\eps^2-0.4027f_{\rm KL}(\eps),
 \label{eq:b1}\\
 b_0(\eps)&=0.06259\eps-0.344f_{\rm KL}(\eps),
 \label{eq:b0}\\
 b_T(\eps)&=0.009307-0.2405\eps-0.03125\eps^2
             -0.06183f_{\rm KL}(5\eps).
 \label{eq:bT}
\end{align}
\end{inputlemma}
This is the final estimate in Section~8.4 of~\cite{SchederFull}, before the substitution $\eps_I=0.029$. Lemma~\ref{lem:epsI-valid} establishes the admissibility of the parameter value used here. We combine this input with the averaged regular estimate developed next.

\section{Analysis of Unique-\texorpdfstring{$3$}{3}-SAT}\label{sec:combination}

\subsection{The regular estimate}

The regular estimate combines the rewards and costs associated with a fixed selected subgraph and then averages the resulting certificates. Fix
$\eps_R=1/10,\ \rho=17/19,\ K=125/109$,
and define
\begin{align}
 c_L&=0.001687\eps_R-0.006378\eps_R^2=0.00010492,\notag\\
 c_T&=0.009307-0.055\eps_R-0.1503f_{\rm KL}(\eps_R).
 \label{eq:regular-net-coefficients}
\end{align}
The entropy bound holds with coefficient $0.006374$; we charge the larger coefficient $0.006378$ when defining the net gain. For each fixed $L$, the rewards, losses, and entropy are evaluated under the same placement law $D_L$. Let $\gain_R^{\rm avg}$ denote the average of the corresponding fixed-selection scores, as defined in Appendix~\ref{app:uniform}.

\begin{proposition}[Averaged regular estimate]\label{prop:regular-averaged}
For every fixed $0<\Thr\le1/1150$, the regular variational certificate satisfies
\begin{equation}\label{eq:regular}
 \gain_R^{\rm avg}\ge
 \rho c_L|G|+0.9\Thr|H_{\rm high}|+c_T|\TwoCC|
 -1.037\rho\eps_R\Thr n.
\end{equation}
For every sufficiently large finite implication strength $w$, uniformly over uniquely satisfiable 3-CNF formulas,
\[
 \Prb[\PPSZ_w(F)=\alpha]
 \ge2^{-p_0n+\gain_R^{\rm avg}-\xi_R(w)n-r_{R,w}(n)},
\]
with the uniform errors specified in Section~\ref{sec:finite-strength}.
\end{proposition}

The proof in Appendix~\ref{app:regular-proof} retains the dependence on $L$ until the final averaging step. Equal marginals give $\E|L|=\rho|G|$. For each fixed selection, the cleanup analysis bounds the forcing loss by $\eps_R\sum_{x\notin\TwoCC}W_x(L)$, while entropy is bounded on the simple support components. Averaging the established variational certificates then uses $\E W_x(L)=\rho W_x(G)$. The same finite-height error bound holds for every selection.

For a fixed root $x$, the forcing loss arises only from pairs of canonical nodes $u,v$ of $T_x$ such that $u$ is a proper ancestor of $v$ and $\{a(u),a(v)\}$ is a support edge of $\underline L$. Grouping these pairs by their support edge, and bounding the relevant descendants of $u$ through the fixed tree $T_{a(u)}$, gives a weight $w_{x,e}\ge0$ for every support pair $e$ of $G$ with
\[
 W_x(L)=\sum_{e\in E(\underline L)}w_{x,e}.
\]
The weights are computed from the full trees before $L$ is chosen; see Figure~\ref{fig:edge-linear-loss} and~\eqref{eq:W-definition}. A selection therefore switches whole families of ancestor--descendant pairs on or off, one family per support edge; this is the linearity used in the averaging step above. The low-edge reward of $x$ uses a different relation: it comes from the support pair of $e_x$, whose endpoints are the children of the root, and it is scaled by the same factor $\rho$.
\begin{figure}[htbp]
 \centering
\begin{tikzpicture}[ppsz picture]
  \coordinate (nx)  at (0,0);
  \coordinate (ny)  at (-2.1,-1.05);
  \coordinate (nz)  at (2.1,-1.05);
  \coordinate (na)  at (-3.15,-2.1);
  \coordinate (nb1) at (-1.05,-2.1);
  \coordinate (nc)  at (1.05,-2.1);
  \coordinate (nd)  at (3.15,-2.1);
  \coordinate (nb2) at (0.55,-3.15);
  \coordinate (nf)  at (1.55,-3.15);

  \draw[ppsz region] (1.62,-.70) -- (2.58,-.70) -- (3.78,-1.95)
    -- (3.78,-2.72) -- (2.25,-3.78) -- (.12,-3.78) -- (.12,-2.62) -- cycle;
  \node[ppsz label, font=\footnotesize, text=black!70, anchor=west]
    at (2.35,-3.32) {$\hookrightarrow T_z$};

  \foreach \p/\c in {nx/ny,nx/nz,ny/na,ny/nb1,nz/nc,nz/nd,nc/nb2,nc/nf}
    \draw[ppsz tree edge] (\p) -- (\c);
  \foreach \v in {na,nb1,nd,nb2,nf}{
    \draw[ppsz stub] ($(\v)+(0,-.07)$) -- ++(-.22,-.36);
    \draw[ppsz stub] ($(\v)+(0,-.07)$) -- ++(.22,-.36);
  }

  \draw[ppsz pair on]  (nx) to[bend right=14] (nb1);
  \draw[ppsz pair on]  (nx) to[bend left=6]   (nb2);
  \draw[ppsz pair off] (nz) to[bend left=16]  (nf);
  \node[ppsz label, text=ppszblue] at (-.22,-1.25) {$e$};
  \node[ppsz label, text=ppszred]  at (2.2,-2.45) {$e'$};

  \foreach \v in {nx,ny,nz,na,nb1,nc,nd,nb2,nf}
    \node[ppsz vertex] at (\v) {};
  \node[ppsz label, above=2pt] at (nx) {$x$};
  \node[ppsz label, anchor=south east] at (ny) {$y$};
  \node[ppsz label, anchor=south west] at (nz) {$z$};
  \node[ppsz label, left=2pt] at (na) {$a$};
  \node[ppsz label, left=2pt] at (nb1) {$b$};
  \node[ppsz label, left=2pt] at (nc) {$c$};
  \node[ppsz label, right=2pt] at (nd) {$d$};
  \node[ppsz label, left=2pt] at (nb2) {$b$};
  \node[ppsz label, right=2.5pt] at (nf) {$f$};

  \foreach \d in {0,1,2,3}
    \node[ppsz label, font=\scriptsize, text=black!55] at (-4.05,-1.05*\d) {$d=\d$};
  \node[ppsz label, font=\footnotesize] at (0,-4.25)
    {(a) canonical nodes of the clause tree $T_x$};

  \begin{scope}[xshift=5.55cm]
    \coordinate (sx) at (0,-.55);
    \coordinate (sb) at (1.75,-.55);
    \coordinate (sy) at (0,-2.25);
    \coordinate (sz) at (1.4,-2.25);
    \coordinate (sf) at (3.0,-2.25);
    \draw[ppsz retained] (sx) -- node[ppsz label, above=1pt] {$e$}
      node[ppsz label, below=2pt, font=\footnotesize] {$w_{x,e}$} (sb);
    \draw[ppsz retained] (sy) -- node[ppsz label, above=1pt] {$e_x$} (sz);
    \draw[ppsz deleted] (sz) -- node[ppsz label, above=1pt] {$e'$}
      node[ppsz label, below=2pt, font=\footnotesize] {$w_{x,e'}$} (sf);
    \foreach \v in {sx,sb,sy,sz,sf} \node[ppsz vertex] at (\v) {};
    \node[ppsz label, left=2pt] at (sx) {$x$};
    \node[ppsz label, right=2pt] at (sb) {$b$};
    \node[ppsz label, left=2pt] at (sy) {$y$};
    \node[ppsz label, below left=2pt] at (sz) {$z$};
    \node[ppsz label, right=2pt] at (sf) {$f$};
    \node[ppsz label, font=\footnotesize] at (1.4,-3.13)
      {$w_{x,e}=c_0\bigl[\tfrac{9}{40}\operatorname{LD}(b,T_x)+\cdots\bigr]$};
    \node[ppsz label, font=\footnotesize] at (1.4,-3.68)
      {$w_{x,e'}=c_0\bigl[2a_1\operatorname{LD}(f,T_z)+\cdots\bigr]$};
    \node[ppsz label, font=\footnotesize] at (1.4,-4.25)
      {(b) support edges of $G$, one fixed $L$};
  \end{scope}
\end{tikzpicture}
 \caption{Edge-linear forcing loss for a root $x$, in a schematic local configuration whose drawn nodes use $C_x=(x\vee\bar y\vee\bar z)$, $C_y=(y\vee\bar a\vee\bar b)$, $C_z=(z\vee\bar c\vee\bar d)$ and $C_c=(c\vee\bar b\vee\bar f)$. Black edges belong to $T_x$; panel (b) shows support pairs of $G$, blue if retained in $L$ and red dashed otherwise. A dotted arrow joins a node to a descendant whose label is adjacent to its own label in $\underline G$, and has the color of that support edge. The edge $e=\{x,b\}$ affects every $b$-labeled node of $T_x$ through the root. The edge $e'=\{z,f\}$ acts only inside the subtree of $z$, whose canonical paths inject label-preservingly into $T_z$. Each family is charged to a label density of a fixed tree, with $c_0=\frac{51}{50}K$; the dots collect the terms of~\eqref{eq:W-definition} from undrawn nodes. Here $e$ contributes $w_{x,e}$ to $W_x(L)$ and $e'$ contributes $0$. The children of the root span $e_x$, which carries the low-edge reward of $x$.}
 \label{fig:edge-linear-loss}
\end{figure}

\subsection{Affine bounds}

We choose the threshold so that the coefficient of the high-edge count cancels after the candidate-edge bound is substituted, and set
\begin{align}
 A&=\rho c_L,\qquad \Thr=\frac{2A}{0.9},\label{eq:AThr}\\
 \Preg&=1.037\rho\eps_R\Thr,\qquad S=c_T-6A.\label{eq:pS}
\end{align}
Substituting~\eqref{eq:candidate-count} into~\eqref{eq:regular} cancels the $H_{\rm high}$ coefficient:
\[
 \gain_R^{\rm avg}\ge A|H|+(c_T-4A)|\TwoCC|-\Preg n.
\]
Together with~\eqref{eq:H-lower}, this gives the regular affine bound; the inherited irregular bound uses the fixed parameter value of Appendix~\ref{app:arithmetic}:
\begin{align}
 \frac{\gain_R^{\rm avg}}n&\ge
 \Lreg(i_0,i_1,\tau):=A-\Preg-2Ai_0-Ai_1+S\tau,\label{eq:Rlinear}\\
 \frac{\gain_I}n&\ge\Lirr(i_0,i_1,\tau):=b_0i_0+b_1i_1+b_T\tau.\label{eq:Ilinear}
\end{align}
The coefficient $6=4+2$ combines the free-edge deletion charge with the conversion from $J_i$ to $\ID_i$.

\subsection{The dual certificate}

The two affine bounds apply to the same input and the same success probability. We obtain a uniform lower bound by allowing the structural coordinates to range over all nonnegative values:
\begin{equation}\label{eq:minimax}
 \Gamma=\inf_{i_0,i_1,\tau\ge0}\max\{\Lreg,\Lirr\}.
\end{equation}
The dual of its epigraph linear program is
\begin{equation}\label{eq:dual-lp}
\begin{aligned}
 \text{maximize}\quad &y_R(A-\Preg)\\
 \text{subject to}\quad &y_R+y_I=1,\qquad y_R,y_I\ge0,\\
 &\begin{aligned}
  -2Ay_R+b_0y_I&\ge0,\\
  -Ay_R+b_1y_I&\ge0,\\
  Sy_R+b_Ty_I&\ge0.
 \end{aligned}
\end{aligned}
\end{equation}
Parameterizing the weights by $y_R=\lambda/(1+\lambda)$ and $y_I=1/(1+\lambda)$ gives the following weighted-average certificate. Its feasibility conditions make every structural-variable coefficient nonnegative.

\begin{proposition}[Affine minimax certificate]\label{prop:certificate}
Suppose $A>\Preg$, $\lambda\ge0$, and
\begin{equation}\label{eq:symbolic-conditions}
 b_0\ge2\lambda A,
 \qquad
 b_1\ge\lambda A,
 \qquad
 b_T+\lambda S\ge0.
\end{equation}
Then, for every $i_0,i_1,\tau\ge0$,
\begin{equation}\label{eq:symbolic-gamma}
 \max\{\Lreg,\Lirr\}
 \ge
 \frac{\lambda(A-\Preg)}{1+\lambda}.
\end{equation}
\end{proposition}

\begin{proof}
Since $\lambda\ge0$,
\begin{align*}
 \max\{\Lreg,\Lirr\}
 &\ge \frac{\lambda\Lreg+\Lirr}{1+\lambda}\\
 &=\frac{1}{1+\lambda}
 \bigl(
   \lambda(A-\Preg)
   +(b_0-2\lambda A)i_0
   +(b_1-\lambda A)i_1
   +(b_T+\lambda S)\tau
 \bigr).
\end{align*}
Every variable coefficient is nonnegative by~\eqref{eq:symbolic-conditions}.
\end{proof}

\subsection{Proof of Theorem~\texorpdfstring{\ref{thm:main}}{1.1}}

We take $\eps_R=1/10$ and the exact irregular parameter value in~\eqref{eq:parameter-values}. The certificate records $0<\Thr<1/1150$, $A>\Preg>0$, $S>0$, $b_0,b_1>0$, and $b_T<0$. Choosing $\lambda=b_1/A$ makes the $i_1$ constraint tight, while
\[
 b_0-2\lambda A=b_0-2b_1>0,\qquad b_T+\lambda S>0.
\]
Proposition~\ref{prop:certificate} gives
\begin{equation}\label{eq:gamma-repair}
 \Gamma\ge\Grep:=\frac{b_1(A-\Preg)}{A+b_1},\qquad
 \Grep-\gamma_{\mathrm{new}}>9.6\cdot10^{-8}.
\end{equation}
Appendix~\ref{app:arithmetic} records the coefficient enclosures, dual weight, and equality profile. At $i_0=\tau=0$ and $i_1=(A-\Preg)/(A+b_1)$, we have $\Lreg=\Lirr=\Grep$ and $0<i_1<1$. Hence the certificate attains the value of this affine relaxation at the fixed parameters. This does not assert that the profile is realized by a SAT instance or that the parameters are globally optimal.

\begin{proof}[Proof of Theorem~\ref{thm:main}]
Fix the positive margin $\Delta=\Grep-\gamma_{\mathrm{new}}>0$. First choose a finite $w_0$ satisfying the first bound below. With this strength fixed, uniformity of the remainders gives a single $n_0\ge4$ such that the second bound holds for every $n\ge n_0$:
\begin{align*}
 \max\{\xi_R(w_0),\xi_I(w_0)\}&<\frac\Delta4,\\
 \max\{r_{R,w_0}(n),r_{I,w_0}(n)\}&<\frac{\Delta n}{4}.
\end{align*}
Both certificates lower-bound the same success probability. Taking the larger certified gain and the larger approximation error gives
\begin{align*}
 \log_2\Prb[\PPSZ_{w_0}(F)=\alpha]
 &\ge-p_0n+\max\{\gain_R^{\rm avg},\gain_I\}
   -\max_{X\in\{R,I\}}\bigl(\xi_X(w_0)n+r_{X,w_0}(n)\bigr)\\
 &\ge-p_0n+\Grep n-\frac{\Delta n}{2}
 \ge-p_0n+\gamma_{\mathrm{new}}n.
\end{align*}
Monotonicity extends the bound to every $w\ge w_0$. For each fixed $w$, one run has polynomial cost, so independent repetition takes $O^*(2^{(p_0-\gamma_{\mathrm{new}})n})$ time. The certified outward comparison $2^{p_0-\gamma_{\mathrm{new}}}<1.306969598$ gives the stated running time.
\end{proof}

\section{Extension to general \texorpdfstring{$3$}{3}-SAT}\label{sec:lifting}

We apply the lifting theorem of Scheder and Steinberger~\cite{SchederSteinberger} to the unique-case guarantee. The quantitative specialization below separates the limiting gain from the gain attainable at a fixed finite implication strength.

\subsection{Lifting at fixed strength}

The two branches of the lifting argument use the following constants and the binary entropy function:
\begin{align}
 p_*&=\frac{2-\log_2 e}{2}=1-\frac{1}{2\ln2},\qquad q_0=p_0-p_*,\label{eq:pstar-q0}\\
 \hbin(\delta)&=-\delta\log_2\delta-(1-\delta)\log_2(1-\delta).\notag
\end{align}
The binary entropy is defined for $0\le\delta\le1$, with the continuous endpoint convention $0\log_2 0=0$.

A variable is liquid at a given stage if each of its two values extends to a satisfying assignment of the residual formula. Let $Q$ be the joint law of a uniform permutation and the satisfying assignment returned by the complete proof heuristic, and let $I(\pi,\alpha)$ count the liquid variables along that run. In the notation of~\cite{SchederSteinberger}, $Q(\pi,\alpha)=(1/n!)2^{-I(\pi,\alpha)}$. The next input gives a success bound when the expected liquid count is large and a restriction to a uniquely satisfiable formula when it is small.

\begin{importedtheorem}[Scheder--Steinberger]\label{thm:SS-imported}
Let $\mathcal P$ be a monotone proof heuristic of error at most $p\ge p_*$ on a formula class closed under restrictions.  Let $Q$ be the distribution induced by the complete proof heuristic and let $I$ be the number of variables that are liquid when processed, in the notation of~\cite{SchederSteinberger}.  For every satisfiable formula $F$ on $n$ variables,
\begin{equation}\label{eq:SS-main}
 \Prb[\operatorname{RandomDecode}(F,\mathcal P)\text{ succeeds}]
 \ge
 2^{-pn+(p-p_*)\E_Q[I]}.
\end{equation}
Moreover, if $\E_Q[I]\le\delta n$, there is a restriction of at most $\delta n$ variables, consistent with a satisfying assignment, whose residual formula is uniquely satisfiable.
\end{importedtheorem}

The probability bound~\eqref{eq:SS-main} is Main Theorem~1.17 of~\cite{SchederSteinberger}; we also prove the restriction consequence as Lemma~\ref{lem:liquid-restriction}. Their Lifting Theorem~1.18 gives the unique-to-general principle. Appendix~\ref{app:lifting-proof} develops the finite-strength specialization, including the conditioning needed to realize a favorable restriction as a PPSZ prefix.

The unique-case branch places a restriction of $\delta n$ variables first, guesses its unforced values correctly, and applies the unique-case bound to the remaining $(1-\delta)n$ variables. For a unique-case bonus $\gamma>0$, this gives the residual gain below. Together with the large-liquid-set contribution $q_0\delta$, it yields the limiting guarantee
\begin{align}
 u_\gamma(\delta)&=\gamma(1-\delta)-(1-p_0)\delta-\hbin(\delta),\label{eq:u-gamma}\\
 \eta_\infty(\gamma)&=\max_{0\le\delta\le1/2}\min\{q_0\delta,u_\gamma(\delta)\}.\label{eq:eta-infty-def}
\end{align}
The term $q_0\delta$ is supplied by~\eqref{eq:SS-main} when $\E_Q[I]\ge\delta n$. The minimum records the guarantee common to the two branches. The following proposition converts this limiting expression into a fixed-strength guarantee.

\begin{proposition}[Quantitative specialization of the lifting theorem]\label{prop:quant-lift}
Suppose that, for some $\gamma>0$, original PPSZ has the following fixed-strength unique-case bound: there are $w_U,m_U$ such that for every $w\ge w_U$ and every uniquely satisfiable 3-CNF formula $G$ on $m\ge m_U$ variables,
\[
 \Prb[\PPSZ_w(G)\text{ finds its unique solution}]
 \ge 2^{-p_0m+\gamma m}.
\]
Then for every $\eta<\eta_\infty(\gamma)$ there are finite $w_G,n_G$ such that, for every $w\ge w_G$, every $n\ge n_G$, and every satisfiable 3-CNF formula $F$ on $n$ variables,
\[
 \Prb[\PPSZ_w(F)\text{ succeeds}]
 \ge 2^{-p_0n+\eta n}.
\]
\end{proposition}

In the proof in Appendix~\ref{app:lifting-proof}, $\eta$ and $\delta$ are fixed first, followed by a sufficiently large finite $w$, and only then is $n\to\infty$ taken. The strict branch margins absorb the linear error $\eps_wn$ in~\eqref{eq:pw-error}.

\Needspace{10\baselineskip}
\subsection{The limiting gain}

For $0<\delta<1/2$, $q_0\delta$ is strictly increasing and $u_\gamma(\delta)$ is strictly decreasing. The maximum in~\eqref{eq:eta-infty-def} is therefore attained at their unique intersection $\delta_\gamma\in(0,1/2)$. The intersection equation, its equivalent entropy form, and the resulting gain are
\begin{align}
 q_0\delta_\gamma&=u_\gamma(\delta_\gamma),\label{eq:intersection}\\
 \hbin(\delta_\gamma)+(1-p_*+\gamma)\delta_\gamma&=\gamma,\label{eq:root-equation}\\
 \eta_\infty(\gamma)&=q_0\delta_\gamma.\label{eq:eta-root}
\end{align}
Implicit differentiation of~\eqref{eq:root-equation} shows that the root is strictly increasing in $\gamma$:
\[
 \frac{d\delta_\gamma}{d\gamma}
 =\frac{1-\delta_\gamma}
 {\log_2((1-\delta_\gamma)/\delta_\gamma)+1-p_*+\gamma}>0.
\]
Consequently, the lifted bonus is strictly increasing in the unique-case bonus.

Appendix~\ref{app:lifting-numerics} gives the numerical instantiation at the theorem bonus. In particular,~\eqref{eq:eta-new-interval} implies $\eta<\eta_\infty(\gamma_{\mathrm{new}})$ for the value in Corollary~\ref{cor:general}, so Proposition~\ref{prop:quant-lift} applies to Theorem~\ref{thm:main}. The same appendix records the rational separator, the strict margin in each branch, and the outward comparison of running-time bases.

\appendix
\section{Proofs for the regular estimate}\label{app:provenance}\label{app:regular-proof}

\subsection{Inherited inputs and virtual clauses}

The detailed $3$-SAT inputs are taken from ECCC Revision~1~\cite{SchederFull}. Definition~31 specifies the augmentation, Lemma~33 gives canonical-tree similarity, and Lemma~34 supplies the degree-two set $H$. For the probability analysis, Corollary~42 bounds the relevant conditional placements and Corollary~50 bounds the cost of bad biases. Sections~7.5--7.7 supply the positive low-edge, high-edge, and $\TwoCC$ contributions; Section~8.4 supplies the irregular coefficients. We use the neighbor-sum form of the cleanup calculation in Section~7.4 to derive our edge-linear forcing-loss bound. The refined entropy estimate and the short-component selection law then give the averaged regular certificate.

\begin{lemma}[Admissibility of the irregular parameter]\label{lem:epsI-valid}
For $0\le\eps_I\le1/5$, every density in the irregular construction of Scheder is nonnegative. This range also implies the restrictions $\eps_I\le4/5$, $\eps_I\le256/600$, and $5\eps_I\le1$ used elsewhere in the source. In particular, the value in~\eqref{eq:parameter-values} is admissible.
\end{lemma}

\begin{proof}
Definition~67 of~\cite{SchederFull} gives, for $0\le r\le1/2$,
\[
 \gamma_{\ID}(r)=10r^2(1-2r)^2,
 \quad
 \gamma_{p\ID}(r)=\frac{61}{6}r^3(1-2r)^2,
 \quad
 \gamma_{\TwoCC}(r)=20r^3(1-2r).
\]
Each function is extended by zero beyond $1/2$. Write their derivatives as $\phi_{\ID}$, $\phi_{p\ID}$, and $\phi_{\TwoCC}$. Substituting $x=2r\in[0,1]$ gives
\begin{align*}
 |\phi_{\ID}(r)|
 &=10x(1-x)|1-2x|\le\frac52,\\
 |\phi_{p\ID}(r)|
 &=\frac{61}{24}x^2(1-x)|3-5x|
 \le\frac{61}{54},\\
 \phi_{\TwoCC}(r)
 &=20r^2(3-8r)\ge-5.
\end{align*}
For a variable outside $\TwoCC$, every derivative in Definition~67 has the form $-a\phi_{\ID}+m\phi_{p\ID}$, with \mbox{$a\in\{0,1\}$} and \mbox{$m\in\{0,1,2\}$}, and is bounded below by
\[
 -\frac52-2\cdot\frac{61}{54}
 =-\frac{257}{54}>-5.
\]
For a variable in $\TwoCC$, the derivative is at least $-5$. Thus every density $1+\eps_I\gamma_v'$ is nonnegative when $\eps_I\le1/5$. The other parameter restrictions in the lemma are weaker, and the chosen value satisfies
$0.07307238160252154687451293138<\frac15$.
\end{proof}

The virtual completion gives every vertex of the canonical digraph exactly two outgoing arcs. The following lemma shows how to use this representation while keeping the original formula class and variable set.

\begin{lemma}[Virtual completion and implication simulation]\label{lem:virtual-simulation}
Let $F$ be a 3-CNF on $n\ge3$ variables, and let $\Fbase$ and $\Faug$ be as in Section~\ref{sec:structure}. These three formulas have the same satisfying assignments on the same variable set. Every clause of $\Faug$ is implied by at most two clauses of $F$. For any partial assignment $\rho$, a value $u$-implied by $\Faug|_\rho$ is $2u$-implied by $F|_\rho$. Consequently, along any fixed satisfying assignment and variable order, $\PPSZ_{2u}(F)$ makes no more guesses than the corresponding strength-$u$ analysis using $\Faug$.
\end{lemma}

\begin{proof}
Let $C$ be a clause on fewer than three variables, and let $Y$ be the chosen set of additional variables. The conjunction of all sign extensions $C\vee D$, with $D$ ranging over the full sign patterns on $Y$, is equivalent to $C$. Indeed, when $C$ is false, the extension whose added literals are all false is also false; when $C$ is true, every extension is true. Each extension is individually implied by $C$. Applying this replacement clause by clause proves that $F$ and $\Fbase$ are equivalent on the same variable set.

Every new clause of $\Faug$ follows from at most two clauses of $\Fbase$, each of which has one original supporting clause. It therefore has at most two original supports, and its addition preserves the satisfying assignments. Entailment is preserved under restriction. An implication using at most $u$ clauses of $\Faug|_\rho$ can consequently be supported by at most $2u$ clauses of $F|_\rho$, after discarding supports satisfied by $\rho$. Along a satisfying run, soundness ensures that every inferred value agrees with that assignment. The strength comparison therefore gives the claimed comparison of guesses.
\end{proof}

Canonical clauses and ordinary canonical trees are chosen from $\Fbase$, whereas $\TwoCC$ is defined using the single-round augmentation $\Faug$. If $n\ge4$, a critical clause of length at most two implied by a pair of base clauses has at least two distinct critical 3-clause extensions, obtained by adding negative literals on unused variables. Its positive variable therefore belongs to $\TwoCC$. This covers the shorter resolvent in the adjacency argument without iterating the closure or increasing the simulation factor.

\subsection{Short components with equal marginals}\label{app:decomposition}

The component-size bound can be enforced directly on paths and cycles. On a path, delete one residue class of edges modulo 18. On a 19-cycle, two deletions at cyclic distances 10 and 9 leave paths with 9 and 8 edges; Figure~\ref{fig:equal-marginal}(a) shows all rotations of this pattern. The next lemma gives the sharp retention bound for the full multigraph class.

\begin{lemma}[Short-component subgraph]\label{lem:short-components}
Every finite loopless multigraph $G$ of maximum degree at most two has a subgraph $L$ whose nontrivial components have at most 17 edges and
\[
 |L|\ge\frac{17}{19}|G|.
\]
The constant $17/19$ is best possible for this graph class.
\end{lemma}

\Needspace{5\baselineskip}
\begin{proof}
Every nontrivial component is a path or a cycle, with a parallel pair regarded as a cycle of length two. On a path, number the edge instances consecutively and delete one residue class modulo 18. Each choice leaves components with at most 17 edges. Averaging over the 18 residues shows that some choice deletes at most a $1/18$ fraction of the edges.

Retain every cycle with at most 17 edges. For a cycle of length $\ell\ge18$, set $k=\lceil\ell/18\rceil$ and choose $k$ deletion positions whose cyclic spacings are at most 18. Such positions are obtained by splitting $\ell$ into $k$ positive parts that differ by at most one. Each remaining path has at most 17 edges. If $k=1$, then $\ell=18$. For $k\ge2$, we have $\ell\ge18k-17$ and
\[
 \frac{k}{\ell}\le\frac{k}{18k-17}\le\frac2{19}.
\]
Combining these choices over the components proves the retention bound. A 19-cycle requires at least two deletions, so the constant is sharp.
\end{proof}

\begin{lemma}[Equal edge and support-pair marginals]\label{lem:edge-marginals}
For every such $G$, there is a finite distribution on subgraphs $L$ with components of at most 17 edges such that every edge instance $e$ satisfies
\[
 \Pr[e\in L]=\rho,\qquad \rho=\frac{17}{19}.
\]
The two instances of a parallel-edge cycle are retained or deleted together. Consequently every support pair of $G$ also has selection probability $\rho$.
\end{lemma}

\begin{proof}
On a path, choose the deleted residue modulo 18 uniformly, so that every edge is retained with probability $17/18$. Keep this candidate with probability $18/19$, taking the empty subgraph otherwise. Each edge then has marginal $\rho$.

On a cycle of length at most 17, retain the whole cycle with probability $\rho$ and delete it otherwise. In particular, the two instances of a parallel pair are retained together. For a longer cycle, choose a uniform rotation of the deletion pattern from Lemma~\ref{lem:short-components}. Each edge is deleted by exactly $k$ of the $\ell$ rotations, so its edge-retention marginal is $q_\ell=1-k/\ell\ge\rho$. Keep the resulting candidate with probability $\rho/q_\ell$, taking the empty graph otherwise. These are finite laws with rational probabilities; their product over the components gives a finite law on $G$.

In a loopless multigraph of maximum degree two, parallel edges form an isolated two-cycle. The coupled selection on that component gives its support pair marginal $\rho$; on every other component, support pairs and edge instances coincide. This proves the support-pair assertion without requiring independence between selected edges.
\end{proof}

In particular, $\E|L|=\rho|G|$. For any nonnegative weights fixed on distinct pairs before selecting $L$,
\begin{equation}\label{eq:marginal-linearity}
 \E_L\sum_{e\in E(\underline L)}w_e
 =\rho\sum_{e\in E(\underline G)}w_e.
\end{equation}
Combining Lemma~\ref{lem:short-components} with~\eqref{eq:candidate-count} would give the valid deterministic count
\[
 \frac{19}{17}|L|+2|H_{\rm high}|+4|\TwoCC|\ge|H|.
\]
We use the equal-marginal law because it controls the averaged loss as well as the retained edge count.

\subsection{Placement distributions and conditioning}\label{app:placements}

Write $T=\TwoCC$ in this appendix and fix $\eps=\eps_R=1/10$. On $[0,1/2]$, put
\[
 \gamma(r)=r(1-2r)^{3/2},\qquad
 \phi(r)=\gamma'(r)=(1-5r)\sqrt{1-2r},
\]
and extend both functions by zero on $(1/2,1]$. We have $\int_0^1\phi=0$ and $-1/\sqrt5\le\phi\le1$.

For each fixed $L$ from Lemma~\ref{lem:edge-marginals}, define the placement law independently across its support components, with component densities
\begin{equation}\label{eq:component-density}
 f_C(z)=1+\eps\sum_{\{a,b\}\in E(\underline C)}\phi(z_a)\phi(z_b).
\end{equation}
These densities integrate to one and are strictly positive: $|E(\underline C)|\le17$ gives
$f_C\ge1-17\eps/\sqrt5>0$. Variables in $T$ are sampled independently with the regular bias of the source
\[
 \gamma_T(r)=40r^{7/2}(1-2r)^2\quad(0\le r\le1/2),
\]
extended by zero outside its stated range. Their density is $1+\eps\gamma_T'(r)$. Every remaining isolated variable is uniform. Together, these choices define the auxiliary law $D_L$ used for both rewards and losses.

Each term in~\eqref{eq:component-density} corresponds to one distinct neighbor pair. Thus a selected parallel pair contributes one correlation term even though its two canonical-clause labels provide two rewards. Its support component is a single edge. Integrating out the other coordinates gives a uniform marginal at each support vertex and density $1+\eps\phi(u)\phi(v)$ at each adjacent pair, as in Corollary~41 of the source.

The cleanup argument uses a restricted family of conditioning events. Fix a rank $r$, and constrain each other placement to one of $\{r\}$, $[0,r]$, $[r,1]$, or $[0,1]$. For these constraints, Corollary~42 of the source gives a conditional lower bound with denominator at least
\[
 1-\frac2{25}\eps(t-1)\ge\frac{109}{125},\qquad t=|E(\underline C)|\le17.
\]
If $T_b$ is the conditional average of $\phi$ on the corresponding interval and $T_b^- =\min(T_b,0)$, this gives
\begin{equation}\label{eq:conditional-loss}
 \Pr_{D_L}[\pi(a)<r\mid\text{these constraints}]
 \ge r+K\eps\gamma(r)\sum_{b:\{a,b\}\in\underline L}T_b^- ,
 \qquad K=\frac{125}{109}.
\end{equation}
Only the component containing $a$ enters this conditional probability. The neighbor sum is nonpositive, so an upper bound $K$ on the reciprocal denominator gives a lower bound on the probability. These are exactly the interval constraints in the first two decorrelation steps of Section~7.3, which may therefore use
\begin{equation}\label{eq:refined-drags}
 \delta_{\rm root}=K\eps\gamma(r)\max\{0,-\phi(r)\},\qquad
 \delta_{\rm nonroot}=K\eps\frac{\gamma(r)^2}{1-r}.
\end{equation}
The remaining decorrelation steps use marginal identities and the repeated-label comparison. They do not use the constant $1.2$, and~\eqref{eq:conditional-loss} is not invoked for their arbitrary exact-value conditioning.

Since $K<1.2$, these drags are no larger than those in the source. The hypotheses of the one-child and $T$-cleanup comparisons (Propositions~51--52) therefore remain satisfied with the same $T$ bias. Corollary~50 continues to bound the effect of a surviving positive drag by $1.02=51/50$: decreasing $\delta_{\max}$ can only decrease $1-Q_{r-\delta_{\max}}$.

\subsection{Edge-linear forcing loss}\label{app:forcing-loss}

We now write the cleanup loss in a form that is linear in the selected support edges. All trees and label densities are the fixed full canonical trees from Section~\ref{sec:structure}. For $d\ge1$, set
\[
 a_d=\frac{(d+1)^{d+1}}{(d+3)^{d+3}}.
\]
For $x\notin T$, define
\begin{equation}\label{eq:W-definition}
\begin{aligned}
 W_x(L)=\frac{51}{50}K\Bigg[&\frac9{40}
       \sum_{b:\{x,b\}\in\underline L}\operatorname{LD}(b,T_x)\\
 &+4\sum_{\substack{u\in\operatorname{Can}(T_x)\\d(u)\ge1}}
     2^{-d(u)}a_{d(u)}
     \sum_{b:\{a(u),b\}\in\underline L}\operatorname{LD}(b,T_{a(u)})\Bigg].
\end{aligned}
\end{equation}
The expression is a nonnegative linear function of support-edge indicators; Figure~\ref{fig:edge-linear-loss} shows a root term and a non-root term. Its weights depend on neither the auxiliary placement sample nor the selected graph $L$. Evaluating the same expression on $G$ defines the edge-weight functional $W_x(G)$; it does not require a placement distribution on the possibly larger components of $G$.

\begin{lemma}[Loss for a fixed selection]\label{lem:fixed-loss}
For every fixed eligible $L$ and every $x\notin T$, the baseline forcing estimate under $D_L$ is
\[
 \Pr_{D_L}[x\text{ is forced}]\ge s_3-\eps W_x(L)-e_h,
\]
where the finite-height error $e_h$ can be made uniformly small over all inputs, roots, and selections as specified in Appendix~\ref{app:uniform}.
\end{lemma}

\begin{proof}
We use the cleanup calculation of Section~7.4 in~\cite{SchederFull}, retaining the individual neighbor sums to track dependence on the selected support edges. For $0<r<1/2$, let
$Q_r=(r/(1-r))^2$ denote the cut probability of the infinite binary tree, and set $Q_r=1$ for $r\ge1/2$. Write $v\prec_L u$ when $u$ is a proper ancestor of $v$ and $\{a(u),a(v)\}\in\underline L$; these are the dotted pairs in Figure~\ref{fig:edge-linear-loss}. During decorrelation, a canonical node $v$ receives a contribution $\delta_{\rm root}$ from a root ancestor and $\delta_{\rm nonroot}$ from each other such ancestor. At most two contributions occur: labels do not repeat along a path, and the support graph has degree at most two. The one-sided cleanup comparisons above remove noncanonical subtrees.

Corollary~50 bounds the damage in the fully cleaned comparison tree by
\[
 \frac{51}{50}\frac{1-2r}{(1-r)^3}
 \sum_{\substack{v\in\operatorname{Can}(T_x)\\v\ne\mathrm{root}}}
 \delta_v r^{d(v)}.
\]
Separate the root contribution from those of the other ancestors. For each fixed canonical ancestor $u$, Lemma~33 of the source injects its canonical descendant paths into $T_{a(u)}$, preserving label sequences. Grouping the resulting terms by neighboring labels bounds the preceding expression by
\begin{align*}
 \frac{51}{50}\Bigg[&
 \frac{\delta_{\rm root}}{r(1-2r)}
 \sum_{b:\{x,b\}\in\underline L}\operatorname{LD}(b,T_x;r)\\
 &+\frac{\delta_{\rm nonroot}}{1-2r}
 \sum_{\substack{u\in\operatorname{Can}(T_x)\\d(u)\ge1}}
 r^{d(u)-1}
 \sum_{b:\{a(u),b\}\in\underline L}\operatorname{LD}(b,T_{a(u)};r)\Bigg].
\end{align*}
For the first term,~\eqref{eq:refined-drags} gives
\[
 \frac{\delta_{\rm root}}{r(1-2r)}
 =K\eps(5r-1)_+(1-2r)\le\frac9{40}K\eps.
\]
The maximum is obtained by completing the square in $(5r-1)(1-2r)$. For the second, if $d\ge1$,
\[
 \frac{\delta_{\rm nonroot}}{1-2r}r^{d-1}
 =K\eps\frac{r^{d+1}(1-2r)^2}{1-r}
 \le4K\eps\,2^{-d}a_d.
\]
To obtain the second bound, use $(1-r)^{-1}\le2$ and substitute $u=2r$; the maximum of $u^{d+1}(1-u)^2$ occurs at $u=(d+1)/(d+3)$. Integrating the nonnegative label-density terms then gives $\eps W_x(L)$. The root marginal is uniform for $x\notin T$, and the biases vanish for $r\ge1/2$. The finite-to-infinite comparison incurs the uniform error $e_h$ specified below. Passing from truncated to full trees can only enlarge the nonnegative loss functional.
\end{proof}

\begin{lemma}[Uniform loss bound and averaging]\label{lem:loss-average}
The depth series and forcing-loss functional satisfy
\[
 \sum_{d\ge1}a_d<0.0545,\qquad
 W_x(G)<1.037\Thr,\qquad
 \E_L W_x(L)=\rho W_x(G)<1.037\rho\Thr.
\]
\end{lemma}

\begin{proof}
The function $(k/(k+2))^k$ decreases for $k>0$, since its logarithmic derivative is
\[
 -\ln(1+2/k)+\frac2{k+2}<0.
\]
The sign follows from $\ln(1+t)>t/(1+t)$ for $t>0$. Together with the integer inequality $7\cdot102^{102}<104^{102}$, monotonicity gives $a_d<1/[7(d+3)^2]$ for $d\ge101$. A telescoping comparison bounds the tail by $1/(7\cdot103)<1/714$. The numerical certificate (Appendix~\ref{app:reproducibility}) checks the finite rational inequality
\[
 \sum_{d=1}^{100}a_d+\frac1{714}<\frac{109}{2000}=0.0545.
\]
The finite comparison and the tail bound prove the stated series estimate.

To bound $W_x(G)$, use the facts that every support vertex has at most two neighbors, both directional densities of each pair are below $\Thr$, and there are at most $2^d$ canonical nodes at depth $d$. These give
\[
 W_x(G)\le\frac{51}{50}K
 \left(\frac9{20}+8\sum_{d\ge1}a_d\right)\Thr
 <\frac{112965}{109000}\Thr<1.037\Thr.
\]
Finally, the weights in~\eqref{eq:W-definition} are fixed before selecting the graph. Applying~\eqref{eq:marginal-linearity} to these support-pair weights gives the exact expectation identity. This uses linearity under the original selection law, without independence between selected edges or conditioning on placement samples.
\end{proof}

\subsection{Rewards}\label{app:rewards}

We next verify that the positive rewards use the same law and baseline as the loss estimate. For a selected low edge instance $e_x=\{y,z\}$, the free-pair convention gives $x,y,z\notin T$. We also need that neither $\{x,y\}$ nor $\{x,z\}$ belongs to $\underline L$.

The following adjacency argument is Proposition~56 of the source, with the short-resolvent case made explicit. An edge $\{x,y\}$ would supply a canonical clause $C_a=(a\vee\bar x\vee\bar y)$. If $a\ne z$, resolving it with $C_x$ gives a distinct critical clause $(a\vee\bar y\vee\bar z)$, contradicting $a\notin T$.

If $a=z$, setting $x=z=0$ must violate a further clause. This is either a second critical clause for $x$ or $z$, or a clause $(x\vee z\vee\bar b)$ with $b\notin\{x,z\}$. In the latter case, resolution with $C_x$ yields $(x\vee\bar y\vee\bar b)$. When $b\ne y$, this is a distinct critical 3-clause. When $b=y$, the critical 2-clause has a distinct critical 3-clause extension because $n\ge4$. Each is a one-round consequence, again placing $x$ or $z$ in $T$. The same reasoning excludes $\{x,z\}$.

It follows by integrating~\eqref{eq:component-density} that the marginal of $(\pi(y),\pi(z))$, conditional on $\pi(x)=r$, remains $1+\eps\phi(y)\phi(z)$. Retain this pair in the partial cleanup of Section~7.5 of the source, and otherwise use the fully cleaned comparison tree from Lemma~\ref{lem:fixed-loss}. The resulting cut-probability difference is at least $\eps\gamma(r)^2(1-Q_r)^2$. Lemma~55 of the source bounds its integral below by $0.001687\eps$. The low-edge reward therefore holds together with the refined baseline loss under $D_L$.

For a high edge instance, the designated dense endpoint lies outside $V(L)$ and $T$. Its placement is consequently uniform and independent of the other coordinates under $D_L$. Section~7.6 of the source compares a partially cleaned tree with the same fully cleaned baseline, adding one-child and repeated-label bonuses. With $f(r)=r(1-2r)/(1-r)^2$ and $g(r)=(1-2r)^2/(1-r)^3$, the corresponding integrands are
\begin{align*}
 \operatorname{OCB}_\delta(d,r)&=\left(f(r)-\frac{2\delta}{1-r}\right)
                 \left(r-\frac{\delta}{1-r}\right)^d,\\
 \operatorname{MLB}_\delta(d,r)&=r g(r)
                 \left(r-\frac{\delta}{1-r}\right)^{d-1}.
\end{align*}
Within the admissible range of the source, all displayed factors are nonnegative, so decreasing $\delta=\delta_{\max}$ does not decrease either bonus. Our shifts are smaller than those used there at $\eps=0.1$. The antichain comparison and Propositions~C.12--C.13 therefore preserve the additional reward $0.9\Thr$ for $\Thr\le1/1150$, on top of the refined baseline in Lemma~\ref{lem:fixed-loss}. Appendix~\ref{app:uniform} controls the truncation of the full high-density tree uniformly.

For $x\in T$, keep the original regular bias and the bound
\[
 \Pr_{D_L}[x\text{ is forced}]\ge s_3+0.009307-0.055\eps-e_h.
\]
The proofs of Theorem~65 and Lemma~66 in the source use the same component conditional bounds and at most two distinct neighbors. Their pessimistic shift $2.4\eps\gamma(r)^2/(1-r)$ remains valid because $2K<2.4$, so the original $T$-reward bound applies.

The rewards are indexed by distinct canonical-clause labels. For a parallel pair, the pair-marginal argument applies to each of its two roots, although the support has one correlation term. Linearity of expectation allows the two forcing rewards to be added without independence of the forcing events. The low, high, and $T$ labels are disjoint, so all rewards and baseline losses can be summed under the same $D_L$.

\subsection{Entropy}\label{app:entropy}

It remains to bound the entropy cost of this placement law. Isolated uniform variables contribute zero. For a nontrivial support component, let $X_v$ be independent uniform coordinates under the reference measure, and write
\[
 Z=\sum_{\{u,v\}\in E(\underline C)}\phi(X_u)\phi(X_v),\qquad
 1\le t=|E(\underline C)|\le17.
\]
With $m_j=\int_0^1\phi(r)^j\,dr$, the substitution $u=\sqrt{1-2r}$ gives
\begin{align*}
 m_j&=2^{-j}\sum_{a=0}^j\binom ja\frac{5^a(-3)^{j-a}}{j+2a+2},\\
 m_1&=0,\qquad m_2=\frac3{32},\qquad m_3=\frac{12}{385},\qquad m_4=\frac9{224}.
\end{align*}
Let $a_C$ count unordered adjacent edge pairs, and let $c_3,c_4$ indicate that the component is respectively a triangle or a four-cycle. Independence of the reference coordinates and $m_1=0$ give
\begin{gather*}
 \E Z=0,\qquad \E Z^2=t m_2^2,\qquad
 \E Z^3=t m_3^2+6c_3m_2^3\ge0,\\
 \begin{aligned}
 \E Z^4={}&t m_4^2+6a_Cm_4m_2^2
 +6\left(\binom t2-a_C\right)m_2^4\\
 &+36c_3m_3^2m_2+24c_4m_2^4.
 \end{aligned}
\end{gather*}
A term in the fourth moment survives exactly when its edge multiplicities give a fourfold edge, two doubled edges, a triangle with one doubled edge, or four distinct edges forming a four-cycle. Every other pattern has a vertex of degree one and vanishes by $m_1=0$. These cases account for all terms in the expansion.

For a cycle of length $t\ge5$, division by $t$ gives the coefficient $3(t-3)\le42$ on $m_2^4$. A path has $a_C=t-1$ and satisfies the same bound because $m_4\ge m_2^2$. For a triangle and a four-cycle, the final contribution per edge is $12m_3^2m_2$ and $9m_2^4$, respectively; each is at most $42m_2^4$ by the stated exact rational comparison. The support of a parallel two-cycle is a single edge, so it is covered by the path calculation with $t=1$. Hence
\begin{equation}\label{eq:fourth-moment}
 \frac{\E Z^4}{t}\le m_4^2+6m_4m_2^2+42m_2^4.
\end{equation}

For $z>-1$,
\[
 (1+z)\ln(1+z)\le z+\frac{z^2}{2}-\frac{z^3}{6}+\frac{z^4}{3}.
\]
For $z\ge0$, the integral remainder in the degree-three Taylor expansion is at most $z^4/12$. For $z=-u\in(-1,0)$, the remaining series is $\sum_{j\ge4}u^j/[j(j-1)]\le u^4/3$. This proves the scalar inequality throughout its stated domain.

Positivity of~\eqref{eq:component-density} permits us to apply this inequality with $z=\eps Z$. The cubic expectation contributes a nonpositive term, which can be discarded in an upper bound. Together with~\eqref{eq:fourth-moment}, this gives
\begin{align}
 \mathrm{KL}_2(D_C\|U_C)
 &\le\frac{\eps^2t}{\ln2}
 \left[\frac{m_2^2}{2}+\frac{\eps^2}{3}
 (m_4^2+6m_4m_2^2+42m_2^4)\right]\notag\\
 &<0.006374\eps^2t.\label{eq:entropy-refinement}
\end{align}
The final comparison uses $\eps=1/10$, the certified bound $\ln2>0.6931$, and a finite rational inequality checked with the numerical certificate (Appendix~\ref{app:reproducibility}). Independence across components makes their entropy costs additive. For the $T$ bias, we retain the source bound $0.1503f_{\rm KL}(\eps)|T|$. Since $|E(\underline L)|\le|L|$, charging the larger coefficient $0.006378$ per retained edge instance is conservative.

\subsection{Averaging and finite strength}\label{app:uniform}

To average the certificates, we need an approximation error that is uniform over the allowed selections. Fix all numerical parameters and the desired error tolerance. For an arbitrary input and its canonical choices, the full trees determine the label densities, classifications, and candidate graph $G$ before $L$ is selected. The truncation heights below depend only on the fixed parameters and tolerance, uniformly over these structural data and over all allowed selections.

At depth $d$, a binary tree has at most $2^d$ nodes. The full label-density tail after depth $a$ is therefore at most
\[
 \mathcal T_a=\int_0^{1/2}\frac{r(1-2r)}{(1-r)^3}(2r)^a\,dr\longrightarrow0.
\]
First choose a core height $a$ with $\mathcal T_a<\Thr/2$ and small enough to meet the required error tolerance. Every full high-density pair then has density at least $\Thr-\mathcal T_a$ in the core. Applying the high-edge proof at this positive threshold, which remains at most $1/1150$, loses at most $0.9\mathcal T_a$ per high edge. The baseline remains the edge-linear bound of Lemma~\ref{lem:fixed-loss}; it is not replaced by a bound at the smaller threshold. Thus the original classification of $G$ is retained.

For low-edge losses, $W_x(L)$ on the full trees dominates the truncated expression term by term, and its weights are not recomputed from $L$. The remaining cleanup approximations use the uniform tail bound from Lemma~48 of the source,
\[
 \min\{1,2^h r^{h-a}\},\qquad 0\le r<1/2,
\]
whose integral tends to zero as the deeper realization height $h\to\infty$ with $a$ fixed. For $r\ge1/2$, the biases vanish and the ordinary finite binary-tree cut probabilities converge uniformly from the endpoint $r=1/2$. The root densities are bounded in terms of the fixed parameters, and every support component has at most 17 edges. Hence the conditioning, cleanup, and finite overlap-pattern bounds are uniform in both $F$ and $L$. No union bound over graph selections is required.

Now choose $h$ so that the total per-variable approximation loss, including the high-density core error, is within the prescribed tolerance. Lemma~8 of the source requires weak strength at most $2^{h+1}$ on the virtual formula, which Lemma~\ref{lem:virtual-simulation} simulates at original strength $2^{h+2}$. Increasing this common strength decreases the forcing error, giving $\xi_R(w)\to0$ uniformly over inputs and selections. The inherited irregular construction has the same uniform finite-tree control, without the graph-selection average. All these losses may be included in $\xi_X(w)n$, with zero additional $r_{X,w}(n)$.

\begin{proof}[Proof of Proposition~\ref{prop:regular-averaged}]
Fix $L$ before sampling placements. Under its law $D_L$, Lemma~\ref{lem:fixed-loss} and Appendix~\ref{app:rewards} give the baseline losses and all three positive reward terms. Subtract the entropy bound~\eqref{eq:entropy-refinement}, using the conservative coefficient in~\eqref{eq:regular-net-coefficients}. The resulting fixed-selection score is
\[
 \mathcal B(L):= c_L|L|+0.9\Thr|H_{\rm high}|+c_T|T|
       -\eps\sum_{x\notin T}W_x(L).
\]
For every fixed $L$,~\eqref{eq:change} bounds the original PPSZ success exponent below by $-p_0n+\mathcal B(L)$ minus the common approximation error. Define $\gain_R^{\rm avg}=\E_L\mathcal B(L)$. Since the probability and error term are the same for every selection, we may use $\max_L \mathcal B(L)$ in the lower bound, and therefore also $\E_L \mathcal B(L)$. Lemmas~\ref{lem:edge-marginals} and~\ref{lem:loss-average} give
\[
 \E_L \mathcal B(L)\ge\rho c_L|G|+0.9\Thr|H_{\rm high}|+c_T|T|
       -1.037\rho\eps\Thr n.
\]
This proves~\eqref{eq:regular} by averaging unconditional variational certificates, each of which already incorporates the entropy cost of its fixed law $D_L$. The support-pair marginals are applied under the unconditioned graph-selection law to weights fixed in advance. Each certificate bounds the same success probability of the original PPSZ algorithm; the graph selection belongs to the analysis.
\end{proof}

\section{Numerical certificates}\label{app:arithmetic}

\subsection{The affine certificate}\label{app:parameters}

All parameter inputs in this appendix are fixed. The regular and irregular parameters are the exact rational values
\begin{equation}\label{eq:parameter-values}
 \eps_R=\frac1{10},\qquad
 \eps_I=0.07307238160252154687451293138.
\end{equation}
The other exact choices are $\rho=17/19$, $K=125/109$, the loss coefficient $1.037$, and the conservative entropy coefficient $0.006378$. Thus $c_L=0.00010492$ exactly. The quantities $A,\Thr,\Preg,S$ are defined in~\eqref{eq:AThr}--\eqref{eq:pS}, using the averaged penalty and $S=c_T-6A$.

The certificate records the following closed outward enclosures, with lower and upper endpoints listed separately. The final row of the second table gives the equality profile at $i_0=\tau=0$.
\begin{center}
\small
\renewcommand{\arraystretch}{1.12}
\begin{tabular}{@{}>{\raggedright\arraybackslash}p{94pt}>{\raggedleft\arraybackslash}p{137pt}>{\raggedleft\arraybackslash}p{137pt}@{}}
\toprule
quantity & \multicolumn{1}{c}{lower endpoint} & \multicolumn{1}{c}{upper endpoint}\\
\midrule
$A$ & $0.00009387578947368421052$ & $0.00009387578947368421053$\\
$\Thr$ & $0.00020861286549707602339$ & $0.00020861286549707602340$\\
$\Preg$ & $0.00001935598003077870113$ & $0.00001935598003077870114$\\
$c_T$ & $0.00302911695303416376704$ & $0.00302911695303416376705$\\
$S$ & $0.00246586221619205850388$ & $0.00246586221619205850389$\\
\bottomrule
\end{tabular}
\end{center}
The irregular coefficients and the resulting dual certificate are:
\begin{center}
\small
\renewcommand{\arraystretch}{1.12}
\begin{tabular}{@{}>{\raggedright\arraybackslash}p{94pt}>{\raggedleft\arraybackslash}p{137pt}>{\raggedleft\arraybackslash}p{137pt}@{}}
\toprule
quantity & \multicolumn{1}{c}{lower endpoint} & \multicolumn{1}{c}{upper endpoint}\\
\midrule
$b_0$ & $0.00363196877285601397166$ & $0.00363196877285601397167$\\
$b_1$ & $0.00114549739595903309354$ & $0.00114549739595903309355$\\
$b_T$ & $-0.01318180201458237930430$ & $-0.01318180201458237930429$\\
$\lambda=b_1/A$ & $12.20226644570744588443510$ & $12.20226644570744588443511$\\
$b_0-2b_1$ & $0.00134097398093794778458$ & $0.00134097398093794778459$\\
$b_T+\lambda S$ & $0.01690730576579577592958$ & $0.01690730576579577592959$\\
$\Grep$ & $0.00006887533849169737751$ & $0.00006887533849169737752$\\
$\Grep-\gamma_{\mathrm{new}}$ & $0.00000009603849169737751$ & $0.00000009603849169737752$\\
$(A-\Preg)/(A+b_1)$ & $0.06012701446085224910548$ & $0.06012701446085224910549$\\
\bottomrule
\end{tabular}
\end{center}
These enclosures establish $A>\Preg>0$, $0<\Thr<1/1150$, the coefficient signs required in Section~\ref{sec:combination}, and the two positive dual margins. They also give the strict slack in~\eqref{eq:gamma-repair}. The limiting certificate value $\Grep$ is kept distinct from the theorem bonus $\gamma_{\mathrm{new}}$, leaving the stated margin for the finite-strength error.

\subsection{The lifting certificate}\label{app:lifting-numerics}

The lifting inputs use the theorem bonus $\gamma_{\mathrm{new}}=0.0000687793$, and $\gamma_{\mathrm{old}}=1/15218$ provides the comparison. The fixed root brackets are
\begin{align*}
 \delta_{\gamma_{\mathrm{old}}}&\in
 [0.00000321978491531273261,0.00000321978491531273262],\\
 \delta_{\gamma_{\mathrm{new}}}&\in
 [0.00000338183369577144614,0.00000338183369577144615].
\end{align*}
The checker verifies opposite signs of~\eqref{eq:root-equation} at the two endpoints of each bracket. The corresponding outward enclosures for the limiting gains and bases are
\begin{align}
 \eta_\infty(\gamma_{\mathrm{old}})&\in[0.00000034658370651684573,0.00000034658370651684574],\label{eq:eta-old-interval}\\
 \eta_\infty(\gamma_{\mathrm{new}})&\in[0.00000036402694215063354,0.00000036402694215063355],\label{eq:eta-new-interval}\\
 2^{p_0-\eta_\infty(\gamma_{\mathrm{old}})}&\in[1.30703159370976208301135,1.30703159370976208301136],\notag\\
 2^{p_0-\eta_\infty(\gamma_{\mathrm{new}})}&\in[1.30703157790679660112482,1.30703157790679660112483].\notag
\end{align}

For the finite-strength guarantee, use the fixed target $\eta=0.000000364$ and rational separator $\delta_0=0.00000338183369$.
At this separator, the two limiting branch margins have enclosures
\begin{align*}
 q_0\delta_0-\eta&\in[0.00000000002694152938422,0.00000000002694152938423],\\
 u_{\gamma_{\mathrm{new}}}(\delta_0)-\eta&\in[0.00000000002705058186497,0.00000000002705058186498].
\end{align*}
Each margin is greater than $2.69\cdot10^{-11}$. Evaluating the bases at the theorem bonuses, rather than at $\Grep$, gives
\begin{align*}
 2^{p_0-\gamma_{\mathrm{new}}}&\in[1.30696959751624698244591,1.30696959751624698244592],\\
 2^{p_0-\eta}&\in[1.30703157793120525344360,1.30703157793120525344361].
\end{align*}
The upper endpoints are below $1.306969598$ and $1.307031578$, respectively.

\begin{proof}[Proof of Corollary~\ref{cor:general}]
Apply Proposition~\ref{prop:quant-lift} to Theorem~\ref{thm:main} at $\gamma=\gamma_{\mathrm{new}}$. Equation~\eqref{eq:eta-new-interval} gives
$\eta=0.000000364 <\eta_\infty(\gamma_{\mathrm{new}})$.
At the fixed rational separator
$\delta_0=0.00000338183369$,
the recorded exact interval bounds place both limiting branches more than $2.69\cdot10^{-11}$ above $\eta$. Proposition~\ref{prop:quant-lift} therefore gives a sufficiently large finite implication strength and the stated success exponent. Finally,
\[
 2^{p_0-0.000000364}
 =1.307031577931205\ldots
 <1.307031578.
\]
\end{proof}

\subsection{Reproducibility}\label{app:reproducibility}

The enclosures in this appendix and the finite rational inequalities used in Lemma~\ref{lem:loss-average} and in~\eqref{eq:entropy-refinement} were checked by a program in exact rational arithmetic. The program verifies only the numerical conditions of the certificate and is not a formal verification of the proofs. The code and the certificate are planned to be made available at \url{https://github.com/jiangxioabai/A-Better-Analysis-For-PPSZ}.

\section{Lifting with finite implication strength}\label{app:lifting-proof}

We prove Proposition~\ref{prop:quant-lift} by combining the large-liquid-set estimate with a restriction to a uniquely satisfiable residual formula. To apply the unique-case guarantee, we must account for the probability of realizing that restriction as a prefix and preserve the distribution of the remaining run. The implication strength is fixed throughout.

\subsection{Liquid variables and unique residual formulas}

Fix a finite strength $w$. By~\eqref{eq:pw-error}, the heuristic $\Pw$ has error at most $p_w=p_0+\eps_w$; put \mbox{$q_w=p_w-p_*=q_0+\eps_w$}.
Applying Theorem~\ref{thm:SS-imported} gives, for every satisfiable 3-CNF formula $F$ on $n$ variables,
\begin{equation}\label{eq:finite-high-I}
 \Prb[\PPSZ_w(F)\text{ succeeds}]
 \ge
 2^{-p_wn+q_w\E_Q[I]}.
\end{equation}
When $\E_Q[I]$ is small, the next lemma supplies the restriction needed for the unique-case branch.

\begin{lemma}[A small liquid set yields a unique residual formula]\label{lem:liquid-restriction}
Suppose $I(\pi,\alpha)\le r$ for a permutation $\pi$ and a satisfying assignment $\alpha$ of $F$.  Let $L$ be the set of variables that are liquid at the moment they are processed along $(\pi,\alpha)$, and let $\rho=\alpha|_L$.  Then $|L|=I(\pi,\alpha)\le r$ and $F|_\rho$ has the unique satisfying assignment $\alpha|_{V\setminus L}$.
\end{lemma}

\begin{proof}
Suppose that a satisfying assignment $\beta$ of $F$ extends $\rho$ and differs from $\alpha$. Let $x$ be the first variable in the order $\pi$ where they differ. Before $x$ is processed, the two assignments agree on every previously processed variable, so both values of $x$ extend to satisfying assignments of the residual formula. Thus $x$ is liquid and belongs to $L$. This contradicts the fact that $\beta$ extends $\rho=\alpha|_L$.
\end{proof}

If $\E_Q[I]\le\delta n$, some pair $(\pi,\alpha)$ in the support of $Q$ satisfies $I(\pi,\alpha)\le\delta n$. Since $I$ is integral, Lemma~\ref{lem:liquid-restriction} gives a restriction of at most $\lfloor\delta n\rfloor$ variables leaving a unique residual solution.

\subsection{Realizing a restriction as a PPSZ prefix}

The restriction must be reached by the original PPSZ run. The next lemma separates the cost of placing its variables first from the cost of guessing their unforced values, and then identifies the conditional law of the suffix.

\begin{lemma}[Prefix realization]\label{lem:prefix-realization}
Let $F$ be a satisfiable CNF formula on variable set $V$, let $\alpha\in\sat(F)$, and let $R\subseteq V$ have size $r$.  Assume that
$F|_{\alpha|_R}$ has the unique satisfying assignment $\alpha|_{V\setminus R}$.  Then, for every finite implication strength $w$,
\begin{equation}\label{eq:prefix-product}
 \Prb[\PPSZ_w(F)=\alpha]
 \ge
 \binom nr^{-1}2^{-r}
 \Prb[\PPSZ_w(F|_{\alpha|_R})=\alpha|_{V\setminus R}].
\end{equation}
\end{lemma}

\begin{proof}
Let $A_R$ be the event that the first $r$ positions of the uniform variable permutation contain exactly the variables in $R$, in any order. Since every $r$-subset is equally likely,
$\Prb[A_R]=\binom nr^{-1}$.
Condition on $A_R$ and on a fixed prefix order $\sigma$ of $R$, and follow the run consistent with $\alpha$. Every value inferred in the prefix agrees with $\alpha$ by soundness. Let $g_\sigma\le r$ count the prefix variables that require guesses. These guesses are independent unbiased bits, so they are all correct with probability
$2^{-g_\sigma}\ge2^{-r}$.
Inferred values therefore improve on the lower bound $2^{-r}$ obtained by charging a guess for every prefix variable.

After a successful prefix, the residual formula is exactly $F|_{\alpha|_R}$. Conditional on $A_R$ and $\sigma$, the relative order on $V\setminus R$ is uniform over all $(n-r)!$ possibilities. For this fixed prefix order, every forcing decision in the prefix is determined before the suffix order is consulted. The event that the prefix guesses are correct is consequently measurable with respect to the prefix order and the random bits consumed there, and is independent of the relative order on $V\setminus R$.

The unused random bits remain independent and unbiased. Thus, conditional on $A_R$, on $\sigma$, and on a successful prefix, the suffix has exactly the law of a fresh run of $\PPSZ_w$ on $F|_{\alpha|_R}$: the proof heuristic is $\Pw$, the residual-variable order is uniform, and the guesses are independent unbiased bits.

The resulting conditional suffix probability is the last factor in~\eqref{eq:prefix-product}. The lower bound is independent of $\sigma$, so averaging over the prefix orders proves the claim.
\end{proof}

If the restriction $\rho$ from Lemma~\ref{lem:liquid-restriction} fixes fewer than $r$ variables, extend it to any $r$-element domain using the unique satisfying assignment of $F|_\rho$. The further restricted formula is still uniquely satisfiable. We may therefore apply Lemma~\ref{lem:prefix-realization} with exactly $r=\lfloor\delta n\rfloor$ variables.

Restrictions preserve the class of 3-CNF formulas, and the heuristic $\Pw$ is the same on the residual formula. The unique-case guarantee therefore applies to the suffix at the same strength $w$ used in the prefix and in~\eqref{eq:finite-high-I}. The theorem applies directly to the residual input; any virtual completion needed for its analysis leaves the algorithmic heuristic $\Pw$ on the actual residual clauses unchanged.

\subsection{Proof of Proposition~\texorpdfstring{\ref{prop:quant-lift}}{4.2}}

Fix $\eta<\eta_\infty(\gamma)$. By the definition of $\eta_\infty$, choose a fixed $\delta\in(0,1/2)$ such that
\begin{equation}\label{eq:strict-lift-slack}
 q_0\delta>\eta,
 \qquad
 u_\gamma(\delta)>\eta.
\end{equation}
Let
$\sigma= \min\{q_0\delta-\eta,\,u_\gamma(\delta)-\eta\}>0$.
First choose a finite implication strength $w\ge w_U$ so large that
\begin{equation}\label{eq:choose-w-lift}
 \eps_w(1-\delta)<\frac{\sigma}{3}.
\end{equation}
The convergence $\eps_w\to0$ permits this choice. Keep $w$ fixed in both branches, and choose $n$ sufficiently large afterward.

Let $F$ be a satisfiable 3-CNF formula on $n$ variables.  There are two cases.

\paragraph{Large-$I$ branch.}
If $\E_Q[I]\ge\delta n$, then~\eqref{eq:finite-high-I} gives
\begin{align*}
 \log_2\Prb[\PPSZ_w(F)\text{ succeeds}]
 &\ge -p_wn+q_w\delta n\\
 &=-p_0n+
 \bigl(q_0\delta-\eps_w(1-\delta)\bigr)n\\
 &\ge -p_0n+(\eta+2\sigma/3)n,
\end{align*}
where the last step uses~\eqref{eq:strict-lift-slack} and~\eqref{eq:choose-w-lift}.

\paragraph{Unique-residual branch.}
Assume $\E_Q[I]<\delta n$. Lemma~\ref{lem:liquid-restriction} gives a restriction of at most $r=\lfloor\delta n\rfloor$ variables leaving a uniquely satisfiable residual formula. Extend it to exactly $r$ variables if necessary, using that unique solution. Put
$\delta_n=r/n,\ m=n-r$.
For all sufficiently large $n$, we have $m\ge m_U$.  Lemma~\ref{lem:prefix-realization} and the unique-case hypothesis give
\begin{align*}
 \Prb[\PPSZ_w(F)\text{ succeeds}]
 &\ge
 \binom nr^{-1}2^{-r}
 2^{-p_0m+\gamma m}\\
 &\ge
 2^{-n\hbin(\delta_n)-r-p_0(n-r)+\gamma(n-r)}\\
 &=
 2^{-p_0n+u_\gamma(\delta_n)n}.
\end{align*}
Here we use $\binom nr\le2^{n\hbin(\delta_n)}$. Since $\delta_n\to\delta$ and $u_\gamma$ is continuous, all sufficiently large $n$ satisfy
\[
 u_\gamma(\delta_n)
 \ge u_\gamma(\delta)-\frac{\sigma}{3}
 \ge\eta+\frac{2\sigma}{3}.
\]

At the chosen strength $w$, both branches give success probability at least $2^{-p_0n+\eta n}$ for every sufficiently large $n$. Set $w_G=w$. For each $w'\ge w_G$, monotonicity gives a coupling under which $\PPSZ_{w'}$ makes no more guesses than $\PPSZ_{w_G}$ along any fixed satisfying assignment. The same lower bound therefore holds for all $w'\ge w_G$. The choices occur in the order
\[
 \gamma
 \longrightarrow
 \eta<\eta_\infty(\gamma)
 \longrightarrow
 \delta
 \longrightarrow
 w
 \longrightarrow
 n,
\]
with the fixed strict margin absorbing the finite-strength error before the limit in $n$ is taken.

\section*{Acknowledgments}
We thank Shiteng Chen for helpful discussions.

\bibliographystyle{alphaurl}
\bibliography{ref}

@article{PPSZ,
  author = {Ramamohan Paturi and Pavel Pudl{\'a}k and Michael E. Saks and Francis Zane},
  title = {An improved exponential-time algorithm for $k$-{SAT}},
  journal = {Journal of the ACM},
  volume = {52}, number = {3}, pages = {337--364}, year = {2005},
  doi = {10.1145/1066100.1066101}
}

@inproceedings{Hertli2011,
  author = {Timon Hertli},
  title = {3-{SAT} Faster and Simpler---{Unique-SAT} Bounds for {PPSZ} Hold in General},
  booktitle = {2011 IEEE 52nd Annual Symposium on Foundations of Computer Science (FOCS)},
  pages = {277--284}, publisher = {IEEE}, year = {2011},
  doi = {10.1109/FOCS.2011.22}
}

@inproceedings{Hertli2014,
  author = {Timon Hertli},
  title = {Breaking the {PPSZ} Barrier for {Unique} 3-{SAT}},
  booktitle = {Automata, Languages, and Programming (ICALP 2014), Part I},
  series = {Lecture Notes in Computer Science}, volume = {8572},
  pages = {600--611}, publisher = {Springer}, year = {2014},
  doi = {10.1007/978-3-662-43948-7_50}
}

@inproceedings{HKZZ2019,
  author = {Thomas Dueholm Hansen and Haim Kaplan and Or Zamir and Uri Zwick},
  title = {Faster $k$-{SAT} algorithms using {Biased-PPSZ}},
  booktitle = {Proceedings of the 51st Annual ACM SIGACT Symposium on Theory of Computing (STOC)},
  pages = {578--589}, publisher = {ACM}, year = {2019},
  doi = {10.1145/3313276.3316359}
}

@article{QinWatanabe2022,
  author = {Tong Qin and Osamu Watanabe},
  title = {An Improvement of the {Biased-PPSZ} Algorithm for the {3SAT} Problem},
  journal = {IEICE Transactions on Information and Systems},
  volume = {E105-D}, number = {3}, pages = {481--490}, year = {2022},
  doi = {10.1587/transinf.2021FCP0009}
}

@inproceedings{SchederFOCS,
  author = {Dominik Scheder},
  title = {{PPSZ} is better than you think},
  booktitle = {2021 IEEE 62nd Annual Symposium on Foundations of Computer Science (FOCS)},
  pages = {205--216}, publisher = {IEEE}, year = {2021},
  doi = {10.1109/FOCS52979.2021.00028}
}

@techreport{SchederFull,
  author = {Dominik Scheder},
  title = {{PPSZ} is better than you think},
  institution = {Electronic Colloquium on Computational Complexity},
  number = {TR21-069}, year = {2021},
  note = {Revision~1, October~15, 2021},
  url = {https://eccc.weizmann.ac.il/report/2021/069/revision/1/download/}
}

@article{SchederJournal,
  author = {Dominik Scheder},
  title = {{PPSZ} is better than you think},
  journal = {TheoretiCS}, volume = {3}, pages = {1--37}, year = {2024},
  note = {Article~5}, doi = {10.46298/theoretics.24.5}
}

@article{SchederSteinberger,
  author = {Dominik Scheder and John Steinberger},
  title = {{PPSZ} for General $k$-{SAT} and {CSP}---Making {Hertli}'s Analysis Simpler and 3-{SAT} Faster},
  journal = {Computational Complexity}, volume = {33}, number = {2},
  pages = {1--48}, year = {2024}, note = {Article~13},
  doi = {10.1007/s00037-024-00259-y}
}

@article{PPZ1999,
  author = {Ramamohan Paturi and Pavel Pudl{\'a}k and Francis Zane},
  title = {Satisfiability Coding Lemma},
  journal = {Chicago Journal of Theoretical Computer Science},
  year = {1999}, note = {Article~11},
  url = {https://cseweb.ucsd.edu/~paturi/myPapers/pubs/PaturiPudlakZane_1999_cjtcs.pdf}
}

@inproceedings{HMS2011,
  author = {Timon Hertli and Robin A. Moser and Dominik Scheder},
  title = {Improving {PPSZ} for 3-{SAT} Using Critical Variables},
  booktitle = {28th International Symposium on Theoretical Aspects of Computer Science (STACS 2011)},
  series = {Leibniz International Proceedings in Informatics}, volume = {9},
  pages = {237--248}, year = {2011},
  publisher = {Schloss Dagstuhl---Leibniz-Zentrum f{\"u}r Informatik},
  doi = {10.4230/LIPIcs.STACS.2011.237}
}

@inproceedings{Schoening1999,
  author = {Uwe Sch{\"o}ning},
  title = {A probabilistic algorithm for $k$-{SAT} and constraint satisfaction problems},
  booktitle = {40th Annual Symposium on Foundations of Computer Science (FOCS)},
  pages = {410--414}, publisher = {IEEE}, year = {1999},
  doi = {10.1109/SFFCS.1999.814612}
}

@inproceedings{IwamaTamaki2004,
  author = {Kazuo Iwama and Suguru Tamaki},
  title = {Improved upper bounds for 3-{SAT}},
  booktitle = {Proceedings of the Fifteenth Annual ACM-SIAM Symposium on Discrete Algorithms (SODA)},
  publisher = {SIAM}, year = {2004},
  note = {Full version: ECCC Report TR03-053}
}

@article{Rolf2006,
  author = {Daniel Rolf},
  title = {Improved Bound for the {PPSZ}/{Sch{\"o}ning}-Algorithm for 3-{SAT}},
  journal = {Journal on Satisfiability, Boolean Modeling and Computation},
  volume = {1}, number = {2}, pages = {111--122}, year = {2006},
  doi = {10.3233/SAT190005}
}

@inproceedings{CSTT2013,
  author = {Shiteng Chen and Dominik Scheder and Navid Talebanfard and Bangsheng Tang},
  title = {Exponential Lower Bounds for the {PPSZ} $k$-{SAT} Algorithm},
  booktitle = {Proceedings of the Twenty-Fourth Annual ACM-SIAM Symposium on Discrete Algorithms (SODA)},
  pages = {1253--1263}, publisher = {SIAM}, year = {2013},
  doi = {10.1137/1.9781611973105.91}
}

@inproceedings{PST2017,
  author = {Pavel Pudl{\'a}k and Dominik Scheder and Navid Talebanfard},
  title = {Tighter Hard Instances for {PPSZ}},
  booktitle = {44th International Colloquium on Automata, Languages, and Programming (ICALP 2017)},
  series = {Leibniz International Proceedings in Informatics}, volume = {80},
  pages = {85:1--85:13}, year = {2017},
  publisher = {Schloss Dagstuhl---Leibniz-Zentrum f{\"u}r Informatik},
  doi = {10.4230/LIPIcs.ICALP.2017.85}
}

@inproceedings{ST2020,
  author = {Dominik Scheder and Navid Talebanfard},
  title = {Super Strong {ETH} Is True for {PPSZ} with Small Resolution Width},
  booktitle = {35th Computational Complexity Conference (CCC 2020)},
  series = {Leibniz International Proceedings in Informatics}, volume = {169},
  pages = {3:1--3:12}, year = {2020},
  publisher = {Schloss Dagstuhl---Leibniz-Zentrum f{\"u}r Informatik},
  doi = {10.4230/LIPIcs.CCC.2020.3}
}

@article{QinWatanabe2020,
  author = {Tong Qin and Osamu Watanabe},
  title = {An improvement of the algorithm of {Hertli} for the unique 3{SAT} problem},
  journal = {Theoretical Computer Science},
  volume = {806}, pages = {70--80}, year = {2020},
  doi = {10.1016/j.tcs.2018.11.023}
}

\end{document}